\documentclass[11pt,a4paper]{article}
\usepackage[top=27mm, bottom=27mm, left=22mm, right=22mm]{geometry}
\usepackage[square, comma, sort&compress, numbers]{natbib}
\usepackage{amsthm,amsmath,amssymb}
\usepackage{bbm}
\usepackage{mathrsfs}
\usepackage{multirow}
\usepackage{amsthm}
\usepackage{microtype}
\usepackage{hyperref}
\usepackage{mathtools}
\usepackage{thm-restate}
\usepackage{thmtools}
\usepackage{ulem}
\usepackage{tabularx}
\usepackage[table]{xcolor}
\usepackage[capitalize]{cleveref}
\newtheorem{theorem}{Theorem}[]
\newtheorem{lemma}[theorem]{Lemma}
\newtheorem{proposition}[theorem]{Proposition}

\newtheorem{conjecture}[theorem]{Conjecture}
\newtheorem{definition}[theorem]{Definition}

\newtheorem{remark}[theorem]{Remark}
\newtheorem{claim}[theorem]{Claim}

\newtheorem{question}[theorem]{Question}

\usepackage{array}
\usepackage{amsmath}

\usepackage{algorithm}
\usepackage{algorithmic}

\usepackage{graphicx}

\usepackage{enumerate}

\usepackage{enumitem}

\begin{document}


\title{When can an expander code correct $\Omega(n)$ errors in $O(n)$ time?\footnote{An extended abstract of this paper has been accepted by Random 2024.}}

\author{Kuan Cheng \footnote{Center on Frontiers of Computing Studies, Peking University, Beijing 100871, China. Email: ckkcdh@pku.edu.cn}
\and Minghui Ouyang\footnote{School of Mathematical Sciences, Peking University, Beijing 100871, China. Email: ouyangminghui1998@gmail.com}
\and Chong Shangguan\footnote{Research Center for Mathematics and Interdisciplinary Sciences, Shandong University, Qingdao 266237, China, and Frontiers Science Center for Nonlinear Expectations, Ministry of Education, Qingdao 266237, China. Email: theoreming@163.com}
\and Yuanting Shen\footnote{Research Center for Mathematics and Interdisciplinary Sciences, Shandong University, Qingdao 266237, China. Email: shenyting121@163.com}
}

\date{}
\maketitle

\begin{abstract}
Tanner codes are graph-based linear codes whose parity-check matrices can be characterized by a bipartite graph $G$ together with a linear inner code $C_0$. Expander codes are Tanner codes whose defining bipartite graph $G$ has good expansion property. This paper is motivated by the following natural and fundamental problem in decoding expander codes:

What are the sufficient and necessary conditions that $\delta$ and $d_0$ must satisfy, so that \textit{every} bipartite expander $G$ with vertex expansion ratio $\delta$ and \textit{every} linear inner code $C_0$ with minimum distance $d_0$ together define an expander code that corrects $\Omega(n)$ errors in $O(n)$ time?

For $C_0$ being the parity-check code, the landmark work of Sipser and Spielman (IEEE-TIT'96) showed that $\delta>3/4$ is sufficient; later Viderman (ACM-TOCT'13) improved this to $\delta>2/3-\Omega(1)$ and he also showed that $\delta>1/2$ is necessary. For general linear code $C_0$, the previously best-known result of Dowling and Gao (IEEE-TIT'18) showed that $d_0=\Omega(c\delta^{-2})$ is sufficient, where $c$ is the left-degree of $G$.

In this paper, we give a near-optimal solution to the above question for general $C_0$ by showing that $\delta d_0>3$ is sufficient and $\delta d_0>1$ is necessary, thereby also significantly improving Dowling-Gao's result. We present two novel algorithms for decoding expander codes, where the first algorithm is deterministic, and the second one is randomized and has a larger decoding radius.
\end{abstract}


\section{Introduction}

Graph-based codes are an important class of error-correcting codes that have received significant attention from both academia and industry. They have a long history in coding theory, dating back to Gallager's \cite{gallager1962low} celebrated \textit{low-density parity-check codes} (LDPC codes for short). LDPC codes are a class of linear codes whose parity-check matrices can be characterized by low-degree (sparse) bipartite graphs, called \textit{factor graphs}. Gallager analyzed the rate and distance of LDPC codes, showing that with high probability, randomly chosen factor graphs give rise to error-correcting codes attaining the Gilbert-Varshamov bound. He also presented an iterative algorithm to decode these codes from errors caused by a binary symmetric channel. Since the 1990s, LDPC codes have received increased attention due to their practical and theoretical performance (see \cite{chung2001design,dimakis2012ldpc,luby2001efficient,mosheiff2020listdecoding,richardson2001design,richardson2001capacity} for a few examples and \cite{guruswami2006iterative} for a survey).

As a generalization of the LDPC codes, Tanner \cite{tanner1981recursive} introduced the so-called \textit{Tanner codes}, as formally defined below. Let $c,d,n$ be positive integers and $L:=[n]$, where $[n]=\{1,\ldots,n\}$. Given a \textit{$(c,d)$-regular} bipartite graph $G$ with bipartition $V(G)=L\cup R$ and a $[d,k_0,d_0]$-linear code $C_0$\footnote{The reader is referred to \cref{sec:notation} for basic definitions on graphs and codes.}, the \textit{Tanner code} $T(G,C_0)\subseteq\mathbb{F}_2^n$ is the collection of all binary vectors $x\in\mathbb{F}_2^n$ with the following property: for every vertex $u\in R$, $x_{N(u)}$ is a codeword of the inner code $C_0$, where $N(u)\subseteq L$ is the set of neighbors of $u$ and $x_{N(u)}=(x_v:v\in N(u))\in\mathbb{F}_2^d$ denotes the length-$d$ subvector of $x$ with coordinates restricted to $N(u)$; in other words,
\begin{align*}
    T(G,C_0):=\{x\in\mathbb{F}_2^n:x_{N(u)}\in C_0 \text{ for every }u\in R\}.
\end{align*}

\textit{Expander codes} are Tanner codes whose defining bipartite graphs have good expansion properties, namely, they are \textit{bipartite expanders}. To be precise, for real numbers $\alpha,\delta\in(0,1]$, a $(c,d)$-regular bipartite graph $G$ with bipartition $V(G)=L\cup R$ with $L=[n]$ is called a \textit{$(c,d,\alpha,\delta)$-bipartite expander} if for each subset $S\subseteq L$ with $|S|\le\alpha n$, $S$ has at least $\delta c|S|$ neighbors in $R$, i.e.,
\begin{align*}
    |N(S)|:=|\cup_{v\in S}N(v)|\ge \delta c|S|.
\end{align*}
As each $S\subseteq L$ can have at most $c|S|$ neighbors in $R$, being a $(c,d,\alpha,\delta)$-bipartite expander means that every bounded size subset in $L$ has as many neighbors in $R$ as possible, up to a constant factor.

Sipser and Spielman \cite{sipser1996expander} studied the Tanner code $T(G,C_0)$ with $G$ being a bipartite expander and $C_0$ being a parity-check code. For simplicity, let $Par=\{(x_1,\ldots,x_d):\sum_{i=1}^d x_i=0\}$ denote the parity-check code in $\mathbb{F}_2^d$. They remarkably showed that the expansion property of $G$ can be used to analyze the minimum distance and the decoding complexity of $T(G,Par)$. Roughly speaking, they showed that for every bipartite expander $G$ with sufficiently large \textit{expansion ratio} $\delta>1/2$, $T(G,Par)$ has minimum distance at least $\alpha n$, which further implies that $T(G,Par)$ defines a class of \textit{asymptotically good} codes. More surprisingly, they showed that if the expansion ratio is even larger, say $\delta>3/4$, then for every such $G$, $T(G,Par)$ admits a linear-time decoding algorithm that corrects a linear number of errors in the adversarial noise model. Soon after, Spielman \cite{spielman1996linear} showed that expander codes can be used to construct asymptotically good codes that can be encoded and decoded both in linear time.

Besides the construction based on vertex expansion, \cite{spielman1996linear} also provides a construction based on spectral expansion.
This construction again inherits the general structure of Tanner code, i.e., it combines of an underlying bipartite graph and an inner code. The main difference is that the underlying graph is an edge-vertex incidence graph of a (non-bipartite) spectral expander.
Spectral expander codes also have linear time encoding and decoding, and their (rate \& distance) parameters are different from those of vertex expander codes.
In this paper, we mainly focus on vertex expander codes. The reader is referred to \cite{ben2008combinatorial,kaufman2018construction,lubotzky1988ramanujan,margulis1973explicit,morgenstern1994existence,reingold2000entropy,skachek2008minimum,DBLP:journals/tit/Zemor01,zemor2001expander} for more details on spectral expanders.

Given the strong performance of expander codes, they have been of particular interest in both coding theory and theoretical computer science, and have been studied extensively throughout the years. For example, \cite{DBLP:conf/stoc/GuruswamiI02,DBLP:journals/tit/RothS06} utilized expander codes to obtain near MDS codes with linear-time decoding. A line of research \cite{arora2012message, chen2023improved, Dowling2018fast,feldman2007LP, DBLP:journals/tit/Ron-ZewiWZ21, viderman2013linear,DBLP:journals/ipl/Viderman13,DBLP:journals/tit/Zemor01} improved the distance analysis and decoding algorithm for expander codes in various settings. Very recently, a sequence of works applied expander codes on quantum LDPC and quantum Tanner code construction, finally achieving asymptotically good constructions and linear-time decoding \cite{DBLP:journals/tit/BreuckmannE21,breuckmann2021quantum, DBLP:conf/stoc/DinurHLV23,DBLP:journals/corr/abs-2004-07935,DBLP:conf/stoc/GuPT23,DBLP:conf/stoc/HastingsHO21,DBLP:conf/stoc/KaufmanT21,DBLP:conf/focs/LeverrierTZ15,DBLP:conf/focs/LeverrierZ22, DBLP:journals/tit/LeverrierZ23, DBLP:conf/soda/LeverrierZ23,  lin2022good,DBLP:conf/stoc/PanteleevK22,  DBLP:journals/tit/PanteleevK22}.

Given the discussion above, it is natural to suspect that the expansion ratio $\delta$ plays a prominent role in analyzing the properties of $T(G,Par)$. More precisely, one can formalize the following question. Note that throughout we always assume that $c,d,\alpha,\delta$ are constants while $n$ tends to infinity.

\begin{question}\label{que-1}
    What is the minimum $\delta>0$ such that every $(c,d,\alpha,\delta)$-bipartite expander $G$ with $V(G)=L\cup R$ and $|L|=n$ defines an expander code $T(G,Par)\subseteq\mathbb{F}_2^n$ that corrects $\Omega_{c,d,\alpha,\delta}(n)$ errors in $O_{c,d,\alpha,\delta}(n)$ time?
\end{question}

This question has already attracted considerable attention. Sipser and Spielman \cite{sipser1996expander} used the bit-flipping algorithm (developed on the original algorithm of Gallager \cite{gallager1962low}) to show that $\delta>3/4$ is sufficient to correct $(2\delta-1)\alpha n$ errors in $O(n)$ time. Using linear programming decoding, Feldman, Malkin, Servedio, Stein and Wainwright \cite{feldman2007LP} showed that $\delta>\frac{2}{3}+\frac{1}{3c}$ sufficient to correct $\frac{3\delta-2}{2\delta-1}\alpha\cdot n$ errors, while at the cost of a ${\rm poly}(n)$ decoding time. Viderman \cite{viderman2013linear} introduced the ``Find Erasures and Decode'' algorithm to show that $\delta>\frac{2}{3}-\frac{1}{6c}$ is sufficient to correct $\Omega(n)$ errors in $O(n)$ time.
Moreover, he also shows that there exists a $(c,d,\alpha,1/2)$-bipartite expander $G$ such that $T(G,Par)$ only has minimum distance two, and therefore cannot correct even one error. Viderman's impossibility result implies that $\delta>1/2$ is necessary for the assertion of \cref{que-1} holding for \textit{every} $(c,d,\alpha,\delta)$-bipartite expander.

The above results only consider the case where the inner code $C_0$ is a parity-check code. Therefore, it is tempting to think about whether one can benefit from a stronger inner code $C_0$. Let us call a code \textit{good} if it can correct $\Omega(n)$ errors in $O(n)$ time. Chilappagari, Nguyen, Vasic and Marcellin \cite{Chilappagari2010trapping} showed that if $G$ has expansion radio $\delta>1/2$ and $C_0$ has minimum distance $d(C_0)\ge\max\{\frac{2}{2\delta-1}-3,2\}$, then every such Tanner code $T(G,C_0)$ is good. The above result implies that for $\epsilon\rightarrow 0$ and $\delta=1/2+\epsilon$, $d(C_0)=\Omega(\epsilon^{-1})$ is sufficient to make every Tanner code $T(G,C_0)$ good. Very recently, Dowling and Gao \cite{Dowling2018fast} significantly relaxed the requirement on $\delta$ by showing that for every $\delta>0$,
\begin{align}\label{eq:Dowling-Gao}
   d(C_0)\ge\Omega(c\delta^{-2})
\end{align}
is sufficient\footnote{More precisely, $d(C_0)\ge2t+c(t-1)^2-1$ with $t>\frac{1}{\delta}$.} to make every Tanner code $T(G,C_0)$ good, and be able to correct $\alpha n$ errors. In particular, their result implies that, as long as the minimum distance of $C_0$ is large enough, any tiny positive expansion ratio is sufficient to construct a good Tanner code.

Putting everything together, it is interesting to understand how the expansion ratio $\delta$ of $G$ and the minimum distance $d_0$ of $C_0$ affect the goodness of the Tanner code. We have the following generalized version of \cref{que-1}.

\begin{question}\label{que-2}
    What are the sufficient and necessary conditions that $\delta$ and $d_0$ must satisfy, so that every $(c,d,\alpha,\delta)$-bipartite expander $G$ with $V(G)=L\cup R, |L|=n$, and every inner linear code $C_0\subseteq\mathbb{F}_2^d$ with $d(C_0)\ge d_0$, together define an expander code $T(G,C_0)\subseteq\mathbb{F}_2^n$ that corrects $\Omega_{c,d,\alpha,\delta}(n)$ errors in $O_{c,d,\alpha,\delta}(n)$ time?
\end{question}

The main purpose of this paper is to provide a near-optimal solution to the above question, as presented in the next subsection. On the negative side, we show that when $d_0 \delta \le 1$, there exists an extension of Viderman's construction, yielding expander codes of constant distance (see \cref{prop:lowerbd} below). Therefore, $d_0 \delta > 1$ is a necessary condition for a general expander code $T(G,C_0)$ to be considered good (compared to the $\delta > \frac{1}{2}$ condition in the case of $T(G,Par)$). On the positive side, we show that $d_0 \delta > 3$ is sufficient to make \textit{every} expander code good (see \cref{thm:main} below for details). Our result only loses a multiplicity by three compared to the above necessary result.

\subsection{Main results}

\paragraph{Deterministic decoding of expander codes.} Our main result, which significantly improves on \eqref{eq:Dowling-Gao}, is presented as follows.

\begin{theorem}\label{thm:main}
    Let $G$ be a $(c,d,\alpha,\delta)$-bipartite expander and $C_0$ be a $[d,k_0,d_0]$-linear code, where $c,d,\alpha,\delta,d_0,k_0$ are positive constants. If $\delta d_0>3$, then there exists a linear-time decoding algorithm for the Tanner code $T(G,C_0)$ that can correct $\gamma n$ errors, where $\gamma=\frac{2\alpha}{d_0(1+0.5c\delta)}$.
\end{theorem}

The above theorem shows that $\delta d_0>3$ is sufficient to make every Tanner code $T(G,C_0)$ good. On the other hand, the next proposition shows that for every $d_0\ge 2$, $\delta d_0>1$ is necessary.

\begin{proposition}\label{prop:lowerbd}
    For every $d,d_0\ge2$ and $n\ge 10d_0$, there exist constants $0<\alpha<1,c\ge3$ and a $(c,d,0.9\alpha,\frac{1}{d_0})$-bipartite expander $G$ with $V(G)=L\cup R$ and $|L|=n$ such that for every $[d,k_0,d_0]$-linear code $C_0$, the Tanner code $T(G,C_0)$ has minimum Hamming distance at most $d_0$.
\end{proposition}

\cref{thm:main} and \cref{prop:lowerbd} together show that our requirement $\delta d_0=\Omega(1)$ is in fact almost optimal for \cref{que-2}. Moreover, we have the following conjecture on the fundamental trade-off between $\delta$ and $d_0$.

\begin{conjecture}\label{con:main}
If $\delta d_0>1$, then for every $(c,d,\alpha,\delta)$-bipartite expander $G$ and every inner code $C_0\subseteq\mathbb{F}_2^d$ with $d(C_0)\ge d_0$, the expander code $T(G,C_0)\subseteq\mathbb{F}_2^n$ can correct $\Omega_{c,d,\alpha,\delta}(n)$ errors in $O_{c,d,\alpha,\delta}(n)$ time.
\end{conjecture}

\paragraph{Randomized decoding of expander codes.} Another important direction in the study of expander codes is to understand the maximum number of errors that can be corrected in a linear-time decoding algorithm. In a recent work, Chen, Cheng, Li, and Ouyang \cite{chen2023improved} obtained a quite satisfactory answer to this problem for $T(G,Par)$. They showed that for every $\delta>1/2$ and $(c,d,\alpha,\delta)$-bipartite expander $G$, $T(G,Par)$ has minimum distance at least $\frac{\alpha}{2(1-\delta)}\cdot n-O(1)$, and this is tight up to a $1-o(1)$ factor. Moreover, for $\delta>\frac{3}{4}$, they also gave a linear-time decoding algorithm which corrects $\frac{3\alpha}{16(1-\delta)}\cdot n$ errors. A similar problem for general expander codes $T(G,C_0)$ was studied by \cite{Dowling2018fast}.

Our decoding algorithm that proves \cref{thm:main} is deterministic and corrects $\gamma n$ errors in linear time. Our next result shows that one can correct more errors by using a randomized algorithm.

\begin{theorem}\label{thm:randDec}
    Let $G$ be a $(c,d,\alpha,\delta)$-bipartite expander and $C_0$ be a $[d,k_0,d_0]$-linear code, where $c,d,\alpha,\delta,d_0,k_0$ are positive constants. If $\delta d_0 > 3$, then there exists a linear-time randomized decoding algorithm for Tanner code $T(G,C_0)$ such that if the input has at most $\alpha n$ errors from a codeword, then with probability $1-\exp\left\{  -\Theta_{c, \delta, d_0}\left( n \right) \right\}$, the decoding algorithm can output the correct codeword.
\end{theorem}

\subsection{Notations and definitions}\label{sec:notation}

A graph is a pair $G = (V,E)$, where $V$ is a set whose elements are called vertices and $E$ is a set of 2-subsets of $V$, whose elements are called edges. For a vertex $u\in V$, the set of \textit{neighbors} of $u$ in $G$ is denoted by $N(u):=\{v\in V:\{u,v\}\in E\}$. For a subset $S\subseteq V(G)$, let $N(S)=\cup_{u\in S} N(u)$ be the set of all the neighbors of the vertices in $S$.
A graph $G$ is \textit{bipartite} if $V(G)$ admits a bipartition $V(G)=L\cup R$ such that both $L$ and $R$ contain no edge. Furthermore, $G$ is $(c,d)$-regular if every vertex $v\in L$ has exactly $c$ neighbors in $R$ and every vertex $u\in R$ has exactly $d$ neighbors in $L$.

Let $\mathbb{F}_2=\{0,1\}$ denote the finite field of size 2. A code $C$ is simply a subset of $\mathbb{F}_2^n$. For two vectors $x=(x_1,\ldots,x_n),~y=(y_1,\ldots,y_n)\in\mathbb{F}_2^n$, the \textit{Hamming distance} between $x$ and $y$, denoted by $d_H(x,y)$, is the number of coordinates where $x$ and $y$ differ, that is, $d_H(x,y)=|\{i\in[n]: x_i\neq y_i\}|$. The \textit{minimum distance} of a code $C\subseteq\mathbb{F}_2^n$, denoted by $d(C)$, is the minimum of $d_H(x,y)$ among all distinct $x,y\in C$. Let wt$(x)$ denote the number of nonzero coordinates of $x$. A code $C\subseteq\mathbb{F}_2^n$ is said to be an \textit{$[n,k,d(C)]$-linear code} if it is a linear subspace in $\mathbb{F}_2^n$ with dimension $k$ and minimum distance $d(C)$. It is well-known that for every linear code $C$, $d(C)=\min\{{\rm wt}(x):x\in C\setminus\{0\}\}$.

Throughout, let $G$ be a $(c,d,\alpha ,\delta)$-bipartite expander, and $C_0$ be a $[d,k_0,d_0]$ linear code. Let $T(G,C_0)$ be the Tanner code defined by $G$ and $C_0$. Let {\rm Check} be the error-detection algorithm of $C_0$, which checks whether a vector in $\mathbb{F}_2^d$ is a codeword of $C_0$.  Assume that {\rm Check} takes $h_0$ time. Similarly, let Decode be the correct-correction algorithm for $C_0$, which corrects up to $\lfloor\frac{d_0-1}{2}\rfloor$ errors. Assume that {\rm Decode} takes $t_0$ time. Note that $h_0,t_0$ are constants depending only on $C_0$ but not on $n$.

Conventionally speaking, let us call the vertices in $L$ \textit{variables} and the vertices in $R$ \textit{constraints}. Given a vector $x\in\mathbb{F}_2^n$, which is corrupted from some codeword $y\in T(G,C_0)$, let us call a constraint $u\in R$ \textit{satisfied} if $x_{N(u)}\in C_0$, otherwise call it \textit{unsatisfied.}

\subsection{Some related works}

Below, we briefly review two previous works \cite{Dowling2018fast, sipser1996expander} that are closely related to our decoding algorithms for Theorems \ref{thm:main} and \ref{thm:randDec}. Let us start from the decoding algorithm of Sipser and Spielman \cite{sipser1996expander}. We summarize as follows the so-called iterated decoding or message-passing algorithm of \cite{sipser1996expander} that decodes $T(G,Par)$.
\begin{itemize}
    \item Let $y\in T(G,Par)$ be the correct codeword that we want to decode from the received vector $x$. In the first round, the algorithm runs ${\rm Check}(x_{N(u)})$ for every $u\in R$. If a constraint $u$ is unsatisfied, then it sends a ``flip'' message to every variable in $N(u)\subseteq L$. Sipser and Spielman showed that as long as the expansion ratio of $G$ is sufficiently large ($\delta>3/4$) and the number of corruptions in $x$ is sufficiently small but not identically zero (that is, $1\le d_H(x,y)\le(2\delta-1)\alpha\cdot n$), then there must exist a variable $v\in L$ that receives $>c/2$ flip messages, which implies that more than half constraints in $N(v)$ are unsatisfied. The algorithm then flips $x_v$ and updates $x$ and the status of the constraints in $N(v)$. Note that since $Par$ is the parity-check code, flipping $x_v$ makes all satisfied constraints in $N(v)$ unsatisfied and all unsatisfied constraints in $N(v)$ satisfied. Therefore, by flipping $x_v$ one can strictly reduce the number of unsatisfied constraints.

    \item The algorithm then runs the above process repeatedly. As long as there are still unsatisfied constraints, the algorithm can find the desired $v\in L$ so that flipping $x_v$ strictly reduces the number of unsatisfied constraints. As there are at most $|R|=cn/d$ unsatisfied constraints, the above process must stop in $O(n)$ rounds and therefore yields an $O(n)$ time decoding algorithm.
\end{itemize}

Dowling and Gao \cite{Dowling2018fast} extend Sipser and Spielman's algorithm from $T(G,Par)$ to the more general setting $T(G,C_0)$ by making use of the minimum distance of $C_0$. Note that their algorithm works for linear codes defined on any finite field, but we will describe it only for $\mathbb{F}_2$.
\begin{itemize}
    \item The algorithm begins by setting a threshold $t\le\lfloor\frac{d_0-1}{2}\rfloor$ and then runs ${\rm Decode}(x_{N(u)})$ for every $u\in R$. If a constraint $u\in R$ satisfies $1\le d_H({\rm Decode}(x_{N(u)}),x_{N(u)})\le t-1$, then it sends a ``flip'' message to every variable $v\in N(u)$ with ${\rm Decode}(x_{N(u)})_v\neq x_v$. Note that ${\rm Decode}(x_{N(u)})\in\mathbb{F}_2^d$ is a codeword in $C_0$. The algorithm then flips all $x_v$ for those $v$ receiving at least one flip, and then updates $x$. Dowling and Gao showed that as long as the minimum distance $d_0$ of $C_0$ is sufficiently large, i.e. it satisfies \eqref{eq:Dowling-Gao}, then flipping all variables that receive at least one flip can reduce the number of corrupted variables in $x$ by some positive fraction.

    \item In the next steps, the algorithm runs the above process repeatedly. As the number of corrupted variables is at most $O(n)$, the algorithm will stop in $O(\log n)$ rounds. Crucially, in order to show that the running time of the algorithm is still linear-order but not of order $n\log n$, the authors proved that the running time of every single round is within a constant factor of the number of corrupted variables at the beginning of this round. As the numbers of corrupted variables form a decreasing geometric sequence with the leading term at most $n$, the total running time, which is within a constant factor of the sum of this geometric sequence, is also $O(n)$.
\end{itemize}

Next, we will summarize some known bounds on the minimum distance of expander codes. Sipser and Speilman \cite{sipser1996expander} showed that when $G$ is a $(c,d,\alpha,\delta)$-bipartite expander with expansion ratio $\delta>1/2$, then the minimum distance of $T(G,Par)$ is at least $\alpha n$. Viderman \cite{viderman2013linear} improved the lower bound above to $2\delta\alpha n$. Very recently, Chen, Cheng, Li and Ouyang \cite{chen2023improved} showed that the optimal minimum distance of $T(G,Par)$ is approximately $\Theta(\frac{\alpha n}{1-\delta})$, thus significantly improving the results of \cite{skachek2008minimum,sipser1996expander,viderman2013linear} when $\delta$ is close to 1. When the inner code $C_0$ is a linear code with minimum distance $d_0$, Dowling and Gao \cite{Dowling2018fast} showed that the expander code $T(G,C_0)$ has minimum Hamming distance at least $d_0\delta\lfloor\alpha n\rfloor$. Note that when $d_0=2$, this recovers Viderman's result.

We will close this subsection by mentioning some explicit constructions of bipartite graphs with good vertex expansion properties, including the lossless expanders and the unique-neighbor expanders. There are several very recent works on the explicit constructions of these expanders, as briefly mentioned below.

Informally speaking, a constant-degree {\it lossless expander} is a $c$-left-regular bipartite graph with an expander radio of $1-\epsilon$ such that $\epsilon$ can be arbitrarily small. Capalbo, Reingold, Vadhan, and Wigderson \cite{capalbo2002randomness} constructed the first infinite family of explicit lossless expanders by using a fairly involved form of the zigzag product. Very recently, Cohen, Roth and Ta-Shma \cite{cohen2023hdx} and independently Golowich \cite{golowich2024new} provided new and simpler explicit constructions of these graphs by combining a {\it large} bipartite spectral expander with a small lossless expander.

A $(c,d)$-regular bipartite graph is called a {\it $(c,d,\alpha,\delta)$-unique-neighbor expander} if for every subset $S\subseteq L$ of size at most $\alpha n$, the number of neighbors of $S$ that have exactly one neighbor in $S$ is at least $\delta c|S|$. It is not hard to see that a $(c,d,\alpha,1-\epsilon)$-bipartite expander is also a $(c,d,\alpha,1-2\epsilon)$-unique-neighbor expander. When $\epsilon<1/2$, explicit bipartite expanders are also unique-neighbor expanders. Other explicit constructions of unique-neighbor expanders can be found in e.g. \cite{alon2002explicit,asherov2024bipartite,becker2016symmetric,hsieh2024explicit,kopparty2023simple}.

\subsection{Key new ideas in our work}\label{sec:idea}

In this subsection, we briefly introduce the key new ideas in our work. Let us focus on the deterministic decoding algorithm that proves \cref{thm:main}. Let us begin by analyzing the following two possible places where the previous algorithm in \cite{Dowling2018fast} could be improved.

In every decoding round of the above algorithm, the constraints in $R$ which satisfy $1\le d_H({\rm Decode}(x_{N(u)}),x_{N(u)})\le t-1$ (and hence send at least one and at most $t-1$ flips to $L$) in fact have two statuses, as detailed below. Let $A$ be the set of constraints $u\in R$ that sends at least one flip and \text{Decode}$(x_{N(u)})$ computes the correct codeword in $C_0$ (i.e., \text{Decode}$(x_{N(u)})=y_{N(u)}$); let $B$ be the set of constraints $u\in R$ that sends at least one flip and \text{Decode}$(x_{N(u)})$ computes an incorrect codeword in $C_0$ (i.e., \text{Decode}$(x_{N(u)})\neq y_{N(u)}$).

\paragraph{Two possible places where the previous algorithm could be improved:}

\begin{itemize}
    \item [(i)] It could be the case that every constraint $u\in A$ satisfies $d_H({\rm Decode}(x_{N(u)}),x_{N(u)})=1$ and hence sends only one correct flip to $L$; in the meanwhile, every constraint $u\in B$ may satisfy $d_H({\rm Decode}(x_{N(u)}),x_{N(u)})=t-1$ and sends as many as $t-1$ flips to $L$, which could be all wrong. In this case, the constraints in $R$ altogether send $|A|$ correct flips and $(t-1)|B|$ wrong flips to the variables in $L$.

    \item [(ii)] Unfortunately, the situation could get even worse. Recall that our bipartite graph $G$ is $(c,d)$-regular. It could be the case that the neighbors of the constraints in $A$ are highly concentrated (e.g., all $|A|$ correct flips are received by as few as $|A|/c$ variables in $L$), and the neighbors of the constraints in $B$ are highly dispersed (e.g., all $(t-1)|B|$ possibly wrong flips are received by as many as $(t-1)|B|$ variables in $L$). As a consequence, a small number of corrupted variables but a large number of correct variables in $L$ receive flip messages.
\end{itemize}

Given the two issues above, if we flip all variables that receive at least one flip, then in the worst case, we could correct $|A|/c$ old corrupt variables but produce $(t-1)|B|$ new corrupt variables. Recall that to make the algorithm in \cite{Dowling2018fast} work, in each round, we need to reduce the number of corrupted variables by at least a positive fraction, which implies that in the worst case it is necessary to have $|A|/c\ge(t-1)|B|$. Together with some lower bound on $|A|$ and upper bound on $|B|$ (see \cite{Dowling2018fast} for details),
one can prove that in such worst scenario \eqref{eq:Dowling-Gao} is necessary for Dowling and Gao's algorithm to work.

Our new algorithm begins by noting that we could indeed fix the two problems mentioned above. To do so, we introduce several new ideas as briefly presented below.

\paragraph{Key new ideas in our work:} Let $F:=\{i\in[n]:x_i\neq y_i\}$ be the set of corrupt variables in $x$. Similarly to \cite{Dowling2018fast}, our new algorithm begins by setting a threshold $t=\lfloor\frac{1}{\delta}\rfloor$ and then runs ${\rm Decode}(x_{N(u)})$ for every $u\in R$.
\begin{itemize}
    \item [(a)] To fix the first problem, if a constraint $u\in R$ satisfies $1\le d_H({\rm Decode}(x_{N(u)}),x_{N(u)})\le t-1$, then instead of sending a flip message to every $v\in N(u)$ with ${\rm Decode}(x_{N(u)})_v\neq x_v$, the new algorithm just arbitrarily picks \textit{exactly one} such variable $v$, and sends a flip message to \textit{only} this specific $v$. By doing so, every constraint in $A\cup B$ sends exactly one flip to $L$.

    \item [(b)] To fix the second problem, we associate each $v\in L$ with a counter $\tau_v\in\{0,1,\ldots,c\}$ that counts the number of flips received by $v$. For each $m\in [c]$, let $S_m$ denote the set of variables that receive exactly $m$ flips. Then, instead of flipping every variable that receives at least one flip, i.e., instead of flipping $\cup_{m=1}^c S_m$, we only flip $S_m$ for some $m\in [c]$. Crucially, we show that if the number $|F|$ of corrupt variables is not too large, then there must \textit{exist} some $m\in[c]$ such that $|S_m|$ has the same order as $|F|$, and more importantly, a $(1/2+\kappa)$-fraction of variables in $|S_m|$ are corrupted (and therefore can be corrected by the flipping operation), where $\kappa$ is an absolute positive constant. Therefore, it follows that by flipping all variables in $S_m$, one can reduce $|F|$ by some positive fraction.
\end{itemize}
Note that the details of (a) and (b) can be found in \cref{sec:easyflip}, where we call the algorithm corresponding to (a) and (b) ``EasyFlip'' and write ${\rm EasyFlip}(x,m)$ as the output of the \hyperref[alg:easyflip]{EasyFlip} if $x$ is the input vector and $S_m$ is flipped (see \cref{alg:easyflip}).

However, there is still a gap that needs to be fixed, that is, how to find the required $S_m$? A plausible solution is to run ${\rm EasyFlip}(x,m)$ for every $m\in [c]$. This would roughly increase the total running time by a $c$ factor, which will still be $O(n)$, provided that the original running time is $O(n)$. Unfortunately, by doing so we still cannot \textit{precisely} identify the required $S_m$, as in general we do not know how to count the number of corrupted variables in some corrupted vector. We will fix this issue by introducing our third key new idea:
\begin{itemize}
    \item [(c)] Note that what we can explicitly count in each round of the algorithm is the number of unsatisfied constraints. Roughly speaking, our strategy is to run \hyperref[alg:easyflip]{EasyFlip} iteratively for a large but still constant number of times and then pick the final output that significantly reduces the number of unsatisfied constraints.

    More precisely, assume that we will run \hyperref[alg:easyflip]{EasyFlip} iteratively for $s$ rounds. Let $x^0:=x$ and write $x^1:={\rm EasyFlip}(x^0,m_1)$ as the output of the 1st \hyperref[alg:easyflip]{EasyFlip} invocation where the variables in $S_{m_1}$ are flipped for some $m_1\in [c]$; more generally, for $k\in[s]$, write $x^k:={\rm EasyFlip}(x^{k-1},m_k)$ as the output of the $k$th \hyperref[alg:easyflip]{EasyFlip} invocation where the variables in $S_{m_k}$ is flipped for some $m_k\in [c]$. Note that in \cref{alg:deepflip} we call the above iterated invocations of \hyperref[alg:easyflip]{EasyFlip} as ``DeepFlip'', and write $x^k:={\rm DeepFlip}(x,(m_1,\ldots,m_k))$ as the output of the $k$th \hyperref[alg:easyflip]{EasyFlip} invocation. For $0\le k\le s$, let $F^k\subseteq L$ and $U^k\subseteq R$ denote the sets of corrupted variables and unsatisfied constraints caused by $x^k$, respectively. We prove that there are constants $0<\epsilon\ll\epsilon'\ll\epsilon''<1$ such that the following two wordy but useful observations hold:

    \item [(c1)] If the number of corrupted variables is reduced \textit{dramatically} then the number of unsatisfied constraints is reduced \textit{significantly}, i.e., if for some $k\in[s]$, $|F^k|\le\epsilon |F^0|$, then $|U^k|\le\epsilon'|U|$;

    \item [(c2)] If the number of unsatisfied constraints is reduced \textit{significantly}, then the number of corrupted variables must be reduced by a least a constant fraction i.e., if for some $k\in[s]$, $|U^k|\le\epsilon' |U^0|$, then $|F^k|\le\epsilon''|F|$.

    In the following, we will briefly argue how we will make use of the two observations (c1) and (c2). Recall that in (b) we have essentially guaranteed that for every $k\in[s]$, there exists some $m^*_k\in[c]$ such that by flipping $S_{m^*_k}$ in \hyperref[alg:easyflip]{EasyFlip}, one could reduce the number of corrupted variables by an $\eta$-fraction for some $\eta\in(0,1)$. It follows that if we run \hyperref[alg:deepflip]{DeepFlip} iteratively for $(m_1,\ldots,m_s)=(m^*_1,\ldots,m^*_s)$, then we have $|F^s|\le(1-\eta)^s|F|<\epsilon |F|$, provided that $s>\log_{(1-\eta)^{-1}} \epsilon^{-1}$ is sufficiently large (but still a constant independent of $n$). Therefore, if we run \hyperref[alg:deepflip]{DeepFlip} thoroughly for all $(m_1,\ldots,m_s)\in [c]^s$, then by (c1) there must exist at least one\footnote{Clearly, ${\rm DeepFlip}(x,(m^*_1,\ldots,m^*_s))$ gives a candidate for such $x^k$.} $x^k:={\rm DeepFlip}(x,(m_1,\ldots,m_k))$ with $k\le s$ such that $|U^k|\le\epsilon'|U|$. Moreover, using the last inequality, such $x^k$ and $(m_1,\ldots,m_k)$ can be explicitly identified. Now, by (c2) we can conclude that the number of corrupted variables is indeed reduced by at least a constant fraction.
\end{itemize}
Note that the above brute-force search only increases the total running time by at most a $c^s$ factor. The details of (c) and the analysis of \hyperref[alg:deepflip]{DeepFlip} can be found in \cref{sec:deepflip} and \cref{alg:deepflip}. Moreover, we call the algorithm that runs ${\rm DeepFlip}(x,(m_1,\ldots,m_s))$ thoroughly for all $(m_1,\ldots,m_s)\in [c]^s$ as ``HardSearch'', and is discussed in \cref{sec:hardsearch} and \cref{alg:hardsearch}. The discussion above basically shows that every \hyperref[alg:hardsearch]{HardSearch} invocation could reduce the number of corrupted variables by a constant fraction.

Running \hyperref[alg:hardsearch]{HardSearch} iteratively for $O(\log n)$ rounds, the total number of corrupted variables will be smaller than $\lfloor\frac{d_0-1}{2}\rfloor$, which can be easily corrected by running Decode for every $u\in R$. The main algorithm that puts everything together is called ``MainDecode'', and is presented in \cref{sec:maindecode} and \cref{alg:main}.

To show that the total running time is still linear in $n$, we adopt an argument similar to that in the previous works (e.g., \cite{Dowling2018fast}). We show that the running time of every \hyperref[alg:hardsearch]{HardSearch} invocation is within a constant factor of the number of corrupted variables at the beginning of this invocation.

Lastly, we would like to mention that our randomized decoding algorithm (see \cref{algo:rand}), which proves \cref{thm:randDec} and has a larger decoding radius than the deterministic algorithm, basically follows from the same framework mentioned above.
Loosely speaking, the high-level idea of the randomized algorithm is to reduce the number of corruptions to a moderate size that can be handled by the deterministic algorithm. For that purpose, we design a random flip strategy which can be summarized as follows.

Recall that for every $m\in[c]$, $S_m$ denotes the set of all variables that receive exactly $m$ flips. First, for every constraint $u\in R$ satisfying $1\le d_H({\rm Decode}(x_{N(u)}),x_{N(u)})\le t$, we arbitrarily pick exactly one variable $v\in N(u)$ with ${\rm Decode}(x_{N(u)})_v\neq x_v$ and send a flip message to this specific $v$. Then, we collect all suspect variables that receive at least one flip. Subsequently, we design a random sampling procedure to select a subset of $\cup_{m\in [c]} S_{m}$ to flip. We show that this procedure can ensure, with high probability, that this subset has more corrupted variables than correct variables, as long as the total number of corruptions is at most $\alpha n$. By applying this strategy iteratively, we can show that in each iteration, the number of corruptions will be reduced by a positive fraction. Then, after running a constant number of iterations, the number of corrupted bits can be reduced to a range that the deterministic algorithm can handle. At this point, the deterministic algorithm is invoked to correct all remaining corrupted variables. Moreover, as $n$ increases, the fail probability of each iteration tends to $0$. Therefore, the overall random algorithm will succeed with high probability.

\subsection{Future research directions}

In this subsection, we list some possible directions for future research.
\begin{itemize}
    \item Perhaps the most attractive problem is trying to solve \cref{que-2} and \cref{con:main}. We have shown that $\delta d_0>3$ is sufficient for \cref{que-2}. However, when $d_0=2$ our result does not recover the best known record $\delta>2/3-\Omega(1)$ of Viderman \cite{viderman2013linear} or the earlier result $\delta>3/4$ of Sipser and Spielman \cite{sipser1996expander}. Therefore, as a first step to solve \cref{que-2} and \cref{con:main} in full generality, it would be interesting to prove that the weaker condition $\delta d_0>4/3$ or even $\delta d_0>3/2$ is also sufficient for \cref{que-2}.

    \item Our decoding algorithm is built upon the earlier works of Sipser and Spielman \cite{sipser1996expander} and Dowling and Gao \cite{Dowling2018fast}. Since the work of Viderman \cite{viderman2013linear} makes a step beyond \cite{sipser1996expander} by introducing the so-called ``Find Erasures and Decode'' algorithm, it would be interesting to know if one can utilize the new idea of \cite{viderman2013linear} to essentially improve our bound $\delta d_0>3$.

    \item Although we showed that when $\delta d_0>3$, every Tanner code $T(G,C_0)$ can correct $\Omega_{c,d,\alpha,\delta}(n)$ errors in $O_{c,d,\alpha,\delta}(n)$ time, we did not make a serious attempt to optimize the hidden constants. It is always interesting to find the optimal constants in $\Omega(\cdot)$ and $O(\cdot)$. In particular, finding the optimal fraction of errors that one can correct in linear time is crucial in decoding expander codes.
\end{itemize}

\section{Collection of some auxiliary lemmas}

Given two subsets $S,T\subseteq V(G)$, let $E(S,T)$ denote the set of edges with one endpoint in $S$ and another endpoint in $T$. 
For every positive integer $t$, let
\begin{itemize}
    \item $N_{\le t}(S)=\{u\in V(G):1\le|N(u)\cap S|\le t\}$,
    \item  $N_{t}(S)=\{u\in V(G):|N(u)\cap S|=t\}$,
    \item and $N_{\ge t}(S)=\{u\in V(G): |N(u)\cap S|\ge t\}$.
\end{itemize}

We will make use of the following crucial property of bipartite expander graphs.

\begin{proposition}[Folklore]\label{prop:expander}
    Let $G$ be a $(c,d,\alpha,\delta)$-bipartite expander. Then, for every set $S\subseteq L$ with $|S|\le\alpha n$ and every integer $t\in[d]$, the following inequality holds
    \begin{align*}
        |N_{\le t}(S)|\ge \frac{\delta(t+1)-1}{t}\cdot c|S|.
    \end{align*}
\end{proposition}

\begin{proof}
    By double counting the number of cross edges between $S$ and $N(S)$, one can infer that
    \begin{align*}
c|S|&=|E(S,N(S))|=\sum\limits_{i=1}^di|N_i(S)|\ge |N_{\le t}(S)|+(t+1)(|N(S)|-|N_{\le t}(S)|)\\
&\ge (t+1)|N(S)|-t|N_{\le t}(S)|\ge (t+1)\delta c|S|-t|N_{\le t}(S)|,
    \end{align*}
    where the last inequality follows by the property of bipartite expanders.
\end{proof}

It was known that for every integer $c,d,n\ge 2$, a random $(c,d)$-regular bipartite graph $G$  with bipartition $V(G)=L\cup R$ and $|L|=n$ satisfied the following property with probability $1-(e/\alpha)^{-\alpha n}$ (see Theorem 26 in \cite{sipser1996expander} or Proposition A.3 in \cite{chen2023improved}). For all $0<\alpha<1$, all subsets $S$ of $L$ with size $\alpha n$ have at least
\begin{align}\label{eq:existence}
n\left(\frac{c}{d}(1-(1-\alpha)^d)-2\alpha\sqrt{c\ln\frac{e}{\alpha}}\right)
\end{align}
neighbors. The following proposition is an easy consequence of the above discussion.

\begin{proposition}\label{prop:random}
     For integers $d,n$ and constant $0<\delta<1$, there exists a $(c,d,\alpha,\delta)$-bipartite expander $G$ with $V(G)=L\cup R$ and $|L|=n$, for some sufficiently small $\alpha$ and sufficiently large $c$.
\end{proposition}

\section{Deterministic decoding: Decoding \texorpdfstring{$\Omega(n)$ corruptions in $O(n)$ time}{}}

\begin{theorem}[restatement of \cref{thm:main}]
    Let $G$ be a $(c,d,\alpha,\delta)$-bipartite expander and $C_0$ be a $[d,k_0,d_0]$-linear code, where $c,d,\alpha,\delta,d_0,k_0$ are positive constants. If $\delta d_0>3$, then there exists a linear-time decoding algorithm for the Tanner code $T(G,C_0)$ which can correct $\gamma n$ errors, where $\gamma=\frac{2\alpha}{d_0(1+0.5c\delta)}$.
\end{theorem}

\begin{remark}
    Our objective in deterministic decoding is not to optimize the decoding radius, but rather the decoding regime. Throughout the algorithm, there may arise situations where the size of corruptions is slightly larger than the initial. To address this, we employ a pruning method (see line $5$ of \cref{alg:deepflip}) within the algorithm to keep the increment under control. However, it is crucial that we keep the initial corruptions smaller than $\gamma n$ to ensure that at any point during the algorithm, the size of corruptions never exceeds the specified parameter $\alpha n$ (see the proof of \cref{lem:deepflip} for details); otherwise, we will have no guarantee regarding the expander property of the set of corruptions.
\end{remark}

We need to set some parameters. Suppose that $d_0 > \frac{3}{\delta} - 1$. Let $t = \lfloor \frac{1}{\delta} \rfloor$. Take $\epsilon_0 > 0$ such that $d_0 > \frac{3}{\delta} - 1 + 2\epsilon_0$ and $\lfloor \frac{1}{\delta} + \epsilon_0 \rfloor = \lfloor \frac{1}{\delta} \rfloor$. For every $0<\epsilon_1< \frac{\epsilon_0 \delta^2}{100} $, let $\epsilon_2=\frac{\epsilon_1}{c+1}\cdot\frac{\delta(t+1)-1}{t}>0$ and
\[ \epsilon_3 = \epsilon_2 \left( 2(1-\epsilon_1) \left( \frac{1}{2} +  \frac{\epsilon_0 \delta^2}{2}  \right) - 1 \right) > 0. \]
It is not hard to check that $\epsilon_1,\epsilon_2$ and $\epsilon_3$ are all well-defined. Lastly, let $\epsilon_4=\frac{\delta d_0-1}{d_0-1}\cdot(1-\epsilon_3)$, $\ell=\left\lceil\log_{1-\epsilon_3}\left(\left\lfloor\frac{d_0-1}{2}\right\rfloor\frac{1}{\gamma n}\right)\right\rceil$ and $s_0=\left\lceil\log_{1-\epsilon_3}\left(\epsilon_4\frac{\delta d_0-1}{d_0-1}\right)\right\rceil$.

\subsection{The main decoding algorithm -- MainDecode}\label{sec:maindecode}

Given a corrupt vector $x\in\mathbb{F}_2^n$ with at most $\gamma n$ corruptions, our main decoding algorithm (see \cref{alg:main} below) works as follows. The algorithm is divided into two parts. In the first part (see steps 2-10 below), it invokes \hyperref[alg:hardsearch]{HardSearch} (see \cref{alg:hardsearch} below) recursively for $\ell$ rounds, where in every round the number of corrupt variables is reduced by a $(1-\epsilon_3)$-fraction. After $\ell$ executions of \hyperref[alg:hardsearch]{HardSearch}, the number of corrupt variables is reduced to at most $\lfloor\frac{d_0-1}{2}\rfloor$. Then, in the second part of the algorithm (see steps 11-13 below), the decoder of the inner code $C_0$ is applied to finish decoding.

\begin{algorithm}[htbp]
\caption{Main decoding algorithm for expander codes -- MainDecode}
\label{alg:main}
\begin{algorithmic}[1]
\REQUIRE $G$, $C_0$, $x\in\mathbbm{F}_2^n$
\ENSURE $x' \in \mathbbm{F}_2^n$
\STATE Set $i=1$ and $x^0=x$
\FOR{$1\le i\le\ell$}
\STATE $U^{i-1}\gets \{u\in R: x^{i-1}_{N(u)}\notin C_0\}$
\IF{$|U^{i-1}|=0$}
\RETURN $x'\gets x^{i-1}$
\ELSE
\STATE $x^{i}\gets$ HardSearch$(x^{i-1})$
\STATE $i\gets i+1$
\ENDIF
\ENDFOR
\FOR{every $u\in R$ such that $x^\ell_{N(u)}\notin C_0$}
\STATE $x^\ell_{N(u)}\gets \text{Decode}(x^\ell_{N(u)})$
\ENDFOR
\RETURN $x'\gets x^\ell$
\end{algorithmic}
\end{algorithm}

The next two lemmas justify the correctness and the linear running time of {\rm \hyperref[alg:main]{MainDecode}}.

\begin{lemma}\label{lem:main}
(i) Let $x$ be the input vector of {\rm \hyperref[alg:hardsearch]{HardSearch}} and let $F$ be the set of corrupt variables of $x$. Let $x':={\rm HardSearch}(x)$ and $F'$ be the set of corrupt variables of $x'$. If $|F|\le\gamma n$, then $|F'|\le(1-\epsilon_3)\cdot|F|$.

(ii) In step 11 of {\rm \hyperref[alg:main]{MainDecode}}, the number of corrupt variables in $x^\ell$ is at most $\lfloor\frac{d_0-1}{2}\rfloor$.
\end{lemma}

\begin{lemma}\label{lem:HardSearchtime}
    (i) Let $x$ be the input vector of {\rm \hyperref[alg:hardsearch]{HardSearch}} and let $F$ be the set of corrupt variables of $x$. If $|F|\le\gamma n$, then the running time of {\rm \hyperref[alg:hardsearch]{HardSearch}} is at most $O(n+|F|)$.

    (ii) Furthermore, if the number of corrupt variables in the input vector of {\rm \hyperref[alg:main]{MainDecode}} is at most $\gamma n$, then the running time of {\rm \hyperref[alg:main]{MainDecode}} is $O(n)$.
\end{lemma}

Assuming the correctness of the above two lemmas, we can prove \cref{thm:main} as follows.

\begin{proof}[Proof of \cref{thm:main}]
    Let $y\in T(G,C_0)$ be a codeword and $x\in\mathbb{F}_2^n$ be a corrupted vector. Let $F=\{i\in[n]:x_i\neq y_i\}$ be the set of corrupt variables of $x$ with respect to $y$. To prove the theorem, it suffices to show that as long as $|F|\le\gamma n$, \hyperref[alg:main]{MainDecode} finds $y$ correctly in linear time. We will analyze the following two cases:
    \begin{itemize}
        \item  If the algorithm returns $x^i$ for some $0\le i\le\ell-1$, then as $|U^i|=0$, we must have $x^i\in T(G,C_0)$. Let $F^i$ be the set of the corrupt variables of $x^i$. Then it follows by \cref{lem:main} (i) that $d(x^i,y)=|F^i|\le(1-\epsilon_3)^i |F|\le(1-\epsilon_3)^i\gamma n<d(T(G,C_0))$, which implies that $x^i=y$.
        \item If the algorithm does not return $x^i$ for any $0\le i\le\ell-1$, then it follows by \cref{lem:main} (ii) that $d(x^{\ell},y)\le\lfloor\frac{d_0-1}{2}\rfloor$. Therefore, one can find $y$ by running Decode for every $u\in R$.
    \end{itemize}
    Moreover, by \cref{lem:HardSearchtime} the running time of \hyperref[alg:main]{MainDecode} is $O(n)$, completing the proof of the theorem.
\end{proof}

The remaining part of this section is organized as follows. In \cref{sec:easyflip} below, we will introduce the basic building block of deterministic decoding -- \hyperref[alg:easyflip]{EasyFlip}, which also corresponds to items (a) and (b) in \cref{sec:idea}. In \cref{sec:deepflip} we will introduce the algorithm \hyperref[alg:deepflip]{DeepFlip}, which runs \hyperref[alg:easyflip]{EasyFlip} iteratively for a constant number of times. \hyperref[alg:deepflip]{DeepFlip} corresponds to item (c) in \cref{sec:idea}. In \cref{sec:hardsearch} we will introduce \hyperref[alg:hardsearch]{HardSearch}, which is designed by running \hyperref[alg:deepflip]{DeepFlip} thoroughly for all choices of $(m_1,\ldots,m_s)$ until the number of unsatisfied constraints is significantly reduced. The proofs of \cref{lem:main} and \cref{lem:HardSearchtime} are also presented in \cref{sec:hardsearch}.

\subsection{The basic building block of deterministic decoding -- EasyFlip}\label{sec:easyflip}

In this subsection, we will present the algorithm \hyperref[alg:easyflip]{EasyFlip} (see \cref{alg:easyflip} below), which is the basic building block of our deterministic decoding. It contains the following two parts:

\begin{itemize}
\item EasyFlip (i): in the first part (see steps 1-6 below), it invokes Decode for each constraint $u\in R$ and sends flips to some variables $v\in L$;
\item EasyFlip (ii): in the second part (see steps 7-11 below), it counts the number of flips received by each variable in $L$ and flips all variables that receive exactly $m$ flips.
\end{itemize}

\begin{algorithm}[htbp]
\caption{Flip all variables receiving exactly $m$ flips -- EasyFlip}
\label{alg:easyflip}
\begin{algorithmic}[1]
\REQUIRE $G$, $C_0$, $x\in\mathbbm{F}_2^n$, and $m\in[c]$
\ENSURE $x'\in\mathbbm{F}_2^n$
\FOR{every $u\in R$}
\STATE $\omega\gets \text{Decode}(x_{N(u)})$
\IF{$1\le d_H(\omega,x_{N(u)})\le t$}
\STATE send a ``flip'' to an arbitrary vertex $v\in N(u)$ with $\omega_v\neq x_v$
\ENDIF
\ENDFOR
\FOR{every $v\in L$}
\IF{$v$ receives exactly $m$ flips}
\STATE flip $x_v$
\ENDIF
\ENDFOR
\RETURN $x'\gets x$
\end{algorithmic}
\end{algorithm}

Our goal is to show that there must exist an integer $m\in[c]$ such that by flipping all variables $v\in L$ that receive exactly $m$ flips, one can reduce the number of corrupt variables in $x'$ by a $(1-\epsilon_3)$-fraction, as compared with $x$. Note that for this moment, it suffices to prove the existence of such an $m$ and we do not need to find it explicitly. In fact, later we will find the required $m$ by exhaustive search.

We make the discussion above precise by the following lemma.

\begin{lemma}\label{lem:easyflip}
  Let $x$ be the input vector of {\rm \hyperref[alg:easyflip]{EasyFlip}} and let $F$ be the set of corrupt variables of $x$. If $|F|\le\alpha n$, then there exists an integer $m\in[c]$ such that the following holds. Let $x'={\rm EasyFlip}(x,m)$ be the output vector of {\rm \hyperref[alg:easyflip]{EasyFlip}} and $F'$ be the set of corrupt variables of $x'$. Then $|F'|\le (1-\epsilon_3)|F|$.
\end{lemma}

The next lemma shows that \hyperref[alg:easyflip]{EasyFlip} runs in linear time.

\begin{lemma}\label{lem:easyflip-time}
    If $|F|\le\alpha n$, then the running time of {\rm \hyperref[alg:easyflip]{EasyFlip}} is at most $O(n+|F|)$ ,  where the hidden constant depends only on $t_0,c,d$.
\end{lemma}

\subsubsection{Proof of \texorpdfstring{\cref{lem:easyflip}}{}}

Let us first introduce some notation and easy inequalities. Let $y\in T(G,C_0)$ be the correct codeword that we want to decode from $x$. Let $A$ be the set of constraints $u\in R$ that sends a flip and \text{Decode}$(x_{N(u)})$ computes the correct codeword in $C_0$ (that is, \text{Decode}$(x_{N(u)})=y_{N(u)}$). Similarly, let $B$ be the set of constraints $u\in R$ that sends a flip and \text{Decode}$(x_{N(u)})$ computes an incorrect codeword in $C_0$ (i.e., \text{Decode}$(x_{N(u)})\neq y_{N(u)}$).

By the definitions of $A$ and $N_{\le t}(F)$, it is easy to see that
\begin{align}\label{eq:def-A}
    A=\{u\in R:1\le|N(u)\cap F|\le t\}=N_{\le t}(F).
\end{align}
Therefore, it follows by \eqref{eq:def-A} and \cref{prop:expander} that
\begin{align}\label{ineq:|A|}
    |A|\ge\frac{\delta(t+1)-1}{t}\cdot c|F|.
\end{align}
Moreover, since a constraint $u\in R$ computes an incorrect codeword in $C_0$ only if it sees at least $d_0-t$ corrupt variables in its neighbors (recall that $d(C_0)\ge d_0$), we have that
\begin{align}\label{eq:def-B}
    B=\{u\in R:|N(u)\cap F|\ge d_0-t\text{ and }\exists\text{ }\omega\in C_0\text{ s.t. } 1\le d_H(\omega,x_{N(u)})\le t\}\subseteq N_{\ge d_0-t}(F).
\end{align}
By counting the number of edges between $F$ and $N(F)$, we see that
\begin{align*}
    (d_0-t)|B|\le|E(F,B)|\le|E(F,N(F))|=c|F|,
\end{align*}
which implies that
\begin{align}\label{ineq:|B|}
    |B|\le\frac{c|F|}{d_0-t}.
\end{align}

Consider the following two equalities,
\begin{align*}
    \sum_{k = 1}^d k \cdot |N_k(F)| =  c |F|~~~~\text{and}~~~~\sum_{k = 1}^d |N_k(F)| \ = \ |N(F)| \ \ge \ \delta c |F|.
\end{align*}

By multiplying the second by $\frac{1}{\delta}+\epsilon_0$ and subtracting the first one, we have
\begin{align*}
    \sum_{k = 1}^t (\frac{1}{\delta} + \epsilon_0 - k) |N_k(F)| - \sum_{k = t+1}^d (k - \frac{1}{\delta} - \epsilon_0) |N_k(F)| \ge \left( \left( \frac{1}{\delta} + \epsilon_0 \right) \delta - 1\right) c |F|\ge \epsilon_0 \delta c |F|,
\end{align*}
Moreover, it follows by \eqref{eq:def-A} and \eqref{eq:def-B} that
\begin{align*}
    &\sum_{k = 1}^t (\frac{1}{\delta} + \epsilon_0 - k) |N_k(F)| - \sum_{k = t+1}^d (k - \frac{1}{\delta} - \epsilon_0) |N_k(F)| \\
    &\le\sum_{k = 1}^t (\frac{1}{\delta} + \epsilon_0 - k) |N_k(F)| - \sum_{k = d_0-t}^d (k - \frac{1}{\delta} - \epsilon_0) |N_k(F)| \\
    &\le (\frac{1}{\delta} + \epsilon_0 - 1) |N_{\le t}(F)| - (d_0 - t - \frac{1}{\delta} - \epsilon_0)|N_{\ge d_0-t}(F)|\\
    &\le (\frac{1}{\delta} + \epsilon_0 - 1) |A| - (d_0 - t - \frac{1}{\delta} - \epsilon_0)|B|.
\end{align*}

As $d_0 > \frac{3}{\delta} - 1 + 2\epsilon_0$ and $t = \lfloor \frac{1}{\delta} \rfloor$, we have $d_0 - t - \frac{1}{\delta} - \epsilon_0 > \frac{1}{\delta} + \epsilon_0 - 1$. Combining the above two inequalities, one can infer that
\begin{align*}
  \epsilon_0 \delta c |F| \le (\frac{1}{\delta} + \epsilon_0 - 1) |A| - (d_0 - t - \frac{1}{\delta} - \epsilon_0)|B| \le (\frac{1}{\delta} + \epsilon_0 - 1) (|A| - |B|) \le \frac{1}{\delta} (|A| - |B|) ,
\end{align*}
which implies that
\begin{align}\label{ineq:compare_A_and_B}
    |A| - |B| \ge \epsilon_0 \delta^2 c |F|.
\end{align}
On the other hand, since $A$ and $B$ are disjoint subsets of $N(F)$, we have that
\begin{align}\label{ineq:trivial_upper_bound_of_A+B}
    |A| + |B| \le |N(F)|\le c |F|.
\end{align}

For every integer $m\in[c]$, let $S_m$ be the set of variables in $L$ that receive exactly $m$ flips. Then the variables in $S_m$ receive a total number of $m|S_m|$ flips. In \hyperref[alg:easyflip]{EasyFlip}, every constraint in $A\cup B$ sends exactly one flip to $L$. The total number of flips sent by constraints in $R$ and received by variables in $L$ is exactly
\begin{align}\label{eq:|A|+|B|}
    |A|+|B|=\sum_{m=1}^c m|S_m|.
\end{align}

Let $Z$ be the set of correct variables that receive at least one flip, i.e., $Z=\left(\cup_{m=1}^cS_m\right)\setminus F$. Observe that the set $F'$ of corrupt variables in the output vector $x'$ consists of corrupt variables not flipped by \hyperref[alg:easyflip]{EasyFlip}, which is $F\setminus S_m$, and correct variables that are erroneously flipped by \hyperref[alg:easyflip]{EasyFlip}, which is $S_m\cap Z$. Therefore,
\begin{align}\label{eq:F'}
    F'=(F\setminus S_m)\cup(S_m\cap Z).
\end{align}

Let $\alpha_m$ be the fraction of corrupt variables in $S_m$. Then we have that
\begin{align}\label{eq:def-alpha}
 \alpha_m=\frac{|S_m\cap F|}{|S_m|}~~\text{and}~~1-\alpha_m=\frac{|S_m\cap Z|}{|S_m|}.
\end{align}
Let $\beta_m$ denote the fraction of flips sent from $A$ to $S_m$ among all flips received by $S_m$, i.e.,
\begin{align}\label{eq:def-beta}
    \beta_m=\frac{\text{the number of flips sent from $A$ to $S_m$}}{m|S_m|}.
\end{align}

The following inequality is crucial in the analysis of \hyperref[alg:easyflip]{EasyFlip}.

\begin{claim}\label{lem:alpha}
For every $m\in[c]$, $\alpha_{m}\ge\beta_{m}$.
\end{claim}

\begin{proof}
As every variable in $S_m$ receives the same number of $m$ flips, by \eqref{eq:def-alpha} the number of flips received by $S_m\setminus F$ is $(1-\alpha_m)m|S_m|$. Moreover, by \eqref{eq:def-beta} the number of flips sent from $B$ to $S_m$ is $(1-\beta_m)m|S_m|$. Since the constraints in $A$ always compute the correct codewords in $C_0$, they always send correct flips to their neighbors in $L$. Therefore, the flips received by $S_m\setminus F$ (which are the wrong flips) must be sent by $B$, which implies that
    \begin{align*}
        (1-\alpha_m)m|S_m|\le(1-\beta_m)m|S_m|,
    \end{align*}
where the inequality follows from the fact that $B$ could also send flips to $S_m\cap F$ (which are the correct flips). Thus, $\alpha_m\ge\beta_m$, as needed.
\end{proof}

The following result shows that there exists an integer $m\in[c]$ such that there exists a \textit{large} set $S_m$ that contains \textit{many} corrupt variables.

\begin{claim}\label{lem:inner01}
If $|F|\le\alpha n$, then there exists an integer $m\in[c]$ such that $\alpha_{m}\ge(1-\epsilon_1)\frac{|A|}{|A|+|B|}$ and $|S_{m}|\ge \epsilon_2|F|$.
\end{claim}

\begin{proof}
Suppose for the sake of contradiction that for every $m\in[c]$, we have either $\alpha_m<(1-\epsilon_1)\frac{|A|}{|A|+|B|}$ or $|S_m|<\epsilon_2|F|$. Then, by counting the number of flips sent from $A$ to $L$ (which is exactly $|A|$), we have that
\begin{align*}
    |A|&=\sum\limits_{m=1}^c\beta_m m|S_m|\le\sum\limits_{m=1}^c\alpha_m m|S_m|\\
    &<(1-\epsilon_1)\frac{|A|}{|A|+|B|}\sum\limits_{m=1}^cm|S_m|+\sum\limits_{m=1}^cm\epsilon_2|F|\\
    &=(1-\epsilon_1)|A|+\frac{c(c+1)}{2}\cdot\epsilon_2|F|,
\end{align*}
where the first inequality follows from \cref{lem:alpha}, the second inequality follows from our assumption on $\alpha_m$ and $|S_m|$, and the last equality follows from \eqref{eq:|A|+|B|}.

Rearranging gives that
\begin{align*}
    |A|<\frac{\epsilon_2(c+1)}{2\epsilon_1}\cdot c|F|=\frac{\delta(t+1)-1}{2t}\cdot c|F|,
\end{align*}
contradicting \eqref{ineq:|A|}.
\end{proof}

Next, we will show that by flipping all the variables in $S_m$, where $m$ satisfies the conclusion of \cref{lem:inner01}, one can reduce the size of the set of corrupt variables by a $(1-\epsilon_3)$-fraction, thereby proving \cref{lem:easyflip}.

\begin{proof}[Proof of \cref{lem:easyflip}]
    Let $m\in[c]$ satisfy the conclusion of \cref{lem:inner01}. 
    Combining the two inequalities \eqref{ineq:compare_A_and_B} and \eqref{ineq:trivial_upper_bound_of_A+B}, one can infer that
    \begin{align}\label{ineq:|A|/(|A|+|B|)}
        \frac{|A|}{|A|+|B|} =\frac{1}{2}+ \frac{|A|-|B|}{2(|A|+|B|)} \ge \frac{1}{2}+ +\frac{ \epsilon_0 \delta^2 c|F|}{2c|F|} = \frac{1}{2} + \frac{\epsilon_0 \delta^2}{2}.
    \end{align}
    Therefore, it follows by \eqref{ineq:|A|/(|A|+|B|)} that
\begin{align}\label{ineq:|A||B|}
    \alpha_m\ge(1-\epsilon_1)\frac{|A|}{|A|+|B|} \ge (1-\epsilon_1) \left(\frac{1}{2} + \frac{\epsilon_0 \delta^2}{2}\right) .
\end{align}
It follows by \eqref{eq:F'} that
\begin{align*}
    |F'|&=(|F|-|S_m\cap F|)+|S_m\cap Z|\\
    &=|F|-(2\alpha_m-1)|S_m|\\
    &\le|F|-\left(2(1-\epsilon_1)\left(\frac{1}{2}+\frac{\epsilon_0 \delta^2}{2}\right)-1\right)|S_m| \\
    &=|F|-(\epsilon_3/\epsilon_2)|S_m|\le|F|-\epsilon_3|F|,
\end{align*}
as needed, where the second equality follows by \eqref{eq:def-alpha}, the first inequality follows by \eqref{ineq:|A||B|}, the last equality follows by the definition of $\epsilon_3$ and the last inequality follows by \cref{lem:inner01}.
\end{proof}

We will conclude by the following inequality, which shows that for an arbitrary $m\in [c]$, flipping $S_m$ would not significantly increase the number of corrupt variables.

\begin{claim}\label{lem:inner03}
    For arbitrary $x\in\mathbb{F}_2^n$ and $m\in[c]$, let $x':={\rm EasyFlip}(x,m)$. Let $F$ and $F'$ be the sets of corrupt variables of $x$ and $x'$, respectively. Then $|F'|\le (1+\frac{c}{d_0-t})|F|$.
\end{claim}

\begin{proof}
    Since the constraints in $A$ always compute the correct codewords in $C_0$, they always send correct flips to their neighbors in $L$. Therefore, the wrong flips must be sent by $B$. Therefore, in the worst case (i.e., assuming that $A=\emptyset$), we have that
    \begin{align*}
        |F'|\le |F|+|B|\le \left(1+\frac{c}{d_0-t}\right)|F|,
    \end{align*}
    where the second inequality follows from \eqref{ineq:|B|}.
\end{proof}

\subsubsection{Proof of \texorpdfstring{\cref{lem:easyflip-time}}{}}

Let $U=\{u\in R:x_{N(u)}\notin C_0\}$ be the set of unsatisfied constraints with respect to $x$. The following inequalities will be useful. First, it follows from \eqref{eq:def-A} and \eqref{eq:def-B} that $A\cup B\subseteq U\subseteq N(F)$. As $A\cap B=\emptyset$, we have that
    \begin{align}\label{ineq:|A|+|B|}
       |A|+|B|\le |U|\le |N(F)|\le c|F|.
    \end{align}
Second, it follows by \eqref{eq:|A|+|B|} that
\begin{align}\label{ineq:|S_m|}
        |S_m|\le\frac{|A|+|B|}{m}\le c|F|.
    \end{align}

Recall that we divided \hyperref[alg:easyflip]{EasyFlip} roughly into two parts, according to the discussion above \cref{alg:easyflip}.  We will compute the running time of these two parts separately.

\paragraph{EasyFlip (i):}
\begin{itemize}
    \item For each constraint $u\in R$, invoking Decode$(x_{N(u)})$ requires $t_0$ time, and computing the Hamming distance $d_H({\rm Decode}(x_{N(u)}),x_{N(u)})$ requires $O(d)$ time. Therefore, this process takes $O((t_0+d)|R|)$ time for all constraints in $R$.

    \item We associate each constraint $u\in R$ with an indicator vector $z^u\in\{0,1\}^d$ that indexes the neighbor of $u$ that receives the flip sent from $u$\footnote{Assume that the $d$ coordinates of $z^u$ are labelled by $d$ the neighbors of $u$.}. Initially, $z^u:=0^d$. If $u$ sends a flip to some $v\in N(u)\subseteq L$, then update $z^u_v:=1$. Note that for every $u\in R$, $z^u$ has at most one nonzero coordinate. Initializing and updating the vectors $z^u$ for all $u\in R$ take $O(d|R|)+O(|A|+|B|)$ time, where we need $O(d|R|)$ time to initialize and $O(|A|+|B|)$ time to update, as $R$ sends $|A|+|B|$ flips to $L$.

    \item In total, EasyFlip (i) takes $O((t_0+d)|R|+|A|+|B|)$ time.
\end{itemize}

\paragraph{EasyFlip (ii):}
\begin{itemize}
    \item We associate each variable $v\in L$ with a counter $\tau_v\in\{0,1,\ldots,c\}$ to count the number of flips received by $v$. Initially, set $\tau_v:=0$ and then update $\tau_v=|\{u\in N(v):z^u_v=1\}|$. The algorithm in fact flips all variables in $S_m=\{v\in L: \tau_v=m\}$. Initializing and updating the counters $\tau_v$ for all $v\in L$ take $O(cn)$ time, as we need $O(cn)$ time to initialize and $O(cn)$ time to update.

    \item Lastly, we need to flip $|S_m|$ variables, which needs $O(|S_m|)$ time.

    \item In total, EasyFlip (ii) takes $O(cn+|S_m|)$ time.
\end{itemize}

To sum up, the running time of \hyperref[alg:easyflip]{EasyFlip} is at most
\begin{align*}
O((t_0+d)|R|+|A|+|B|)+O(cn+|S_m|)=O\big((t_0/d+1)cn+c|F|\big)=O(n+|F|),
\end{align*}
where the first equality follows from $|R|=\frac{cn}{d}$, \eqref{ineq:|A|+|B|} and \eqref{ineq:|S_m|}.

\subsection{Running EasyFlip iteratively for a constant number of times -- DeepFlip}\label{sec:deepflip}

In this subsection, we will present and analyze \hyperref[alg:deepflip]{DeepFlip} (see \cref{alg:deepflip} below), which is designed by running \hyperref[alg:easyflip]{EasyFlip} iteratively for $s$ times for a particular choice of $(m_1,\ldots,m_s)\in [c]^s$. Note that by iteratively we mean a sequence of operations $x^0:=x,x^1:={\rm EasyFlip}(x^0,m_1),\ldots,x^s={\rm EasyFlip}(x^{s-1},m_s)$.

\begin{algorithm}[htbp]
\caption{Running EasyFlip iteratively for a particular choice $(m_1,\ldots,m_s)\in [c]^s$ -- DeepFlip}
\label{alg:deepflip}
\begin{algorithmic}[1]
\REQUIRE $G$, $C_0$, $x\in\mathbbm{F}_2^n$, and $(m_1,\ldots,m_s)\in [c]^s$
\ENSURE $x^s\in\mathbbm{F}_2^n$ or $\bot$
\STATE Set $k=1$ and $x^0=x$
\FOR{$1\le k\le s$}
\STATE $x^{k}\gets$ EasyFlip$(x^{k-1},m_{k})$
\STATE $U^{k}\gets\{u\in R:x^{k}_{N(u)}\notin C_0\}$
\IF{$|U^{k}|>(1-\epsilon_3)^{k}\cdot c\gamma n$}
\RETURN  $\bot$
\ELSE
\STATE $k\gets k+1$
\ENDIF
\ENDFOR
\RETURN $x^s$
\end{algorithmic}
\end{algorithm}

Our goal is to show that as long as the number of corrupt variables in $x$ is not too large, by running \hyperref[alg:easyflip]{EasyFlip} iteratively for a large enough (but still constant) number of times, there exists a vector $(m_1,\ldots,m_s)\in[c]^s$ such that the number of corrupt variables in the final output $x^s$ is at most a $(1-\epsilon_3)$-fraction of the number of corrupt variables in the initial input $x$. Most importantly, later we will show that such a vector $(m_1,\ldots,m_s)$ can be found \textit{explicitly} and \textit{efficiently}.

The above assertion will be made precise by the following lemma.

\begin{lemma}\label{lem:deepflip}
  Let $x$ be the input vector of \hyperref[alg:deepflip]{DeepFlip} and let $F$ be the set of corrupt variables of $x$. If $|F|\le\gamma n$, then for every $s\ge s_0$ there exists a nonempty subset $M\subseteq [c]^s$ such that the following holds for every $(m_1,\ldots,m_s)\in M$. Let $x^s:=${\rm DeepFlip}$(x,(m_1,\ldots,m_s))$ be the output vector of {\rm \hyperref[alg:deepflip]{DeepFlip}} and $F^s$ be the set of corrupt variables of $x^s$. Then $|F^s|\le (1-\epsilon_3)|F|$.
\end{lemma}

The next lemma shows that for every fixed $(m_1,\ldots,m_s)\in[c]^s$, the algorithm \hyperref[alg:deepflip]{DeepFlip} runs in linear time.

\begin{lemma}\label{lem:deepflip-time}
    If $|F|\le\gamma n$ and $s$ is a constant, then the running time of {\rm \hyperref[alg:deepflip]{DeepFlip}} is at most $O(n+|F|)$, where the hidden constant depends only on $t_0,c,d,s$.
\end{lemma}

\subsubsection{Proof of \texorpdfstring{\cref{lem:deepflip}}{}}

Given $(m_1,\ldots,m_s)\in [c]^s$ and $x^0:=x$, for each $k\in[s]$, let $x^k:=${\rm EasyFlip}$(x^{k-1},m_k)$. With this notation,
\begin{align*}
    x^s={\rm EasyFlip}(x^{s-1},m_s)={\rm DeepFlip}(x,(m_1,\ldots,m_s)),
\end{align*}
is exactly the output vector of \hyperref[alg:deepflip]{DeepFlip}. Let $F$ be the set of corrupt variables in $x$ and $U$ be the set of unsatisfied constraints with respect to $x$. Sometimes, we will also use $F^0:=F$ and $U^0:=U$. For $k\in[s]$, define $F^k$ and $U^k$ similarly with $x$ replaced by $x^k$. Then
\begin{align}\label{eq:unsatisfied_constraints}
    N_{\le d_0-1}(F)\subseteq U\subseteq N(F),
\end{align}
where the first inclusion holds since $d(C_0)=d_0$.

The following lemma can be viewed as an ``idealized'' version of \cref{lem:deepflip}.

\begin{lemma}\label{lem:deepflip-ideal}
    With the above notation, the following holds.
    If $|F|\le \alpha n$, then there exists a vector $(m_1,\ldots,m_s)\in[c]^s$ such that
    \begin{itemize}
        \item [{\rm (i)}] $|F^s|\le(1-\epsilon_3)^s|F|$;
        \item [{\rm (ii)}] for each $k\in[s]$, $|U^k|\le (1-\epsilon_3)^k\cdot c|F|$;
        \item [{\rm (iii)}] $|U^s|\le(1-\epsilon_3)^s\cdot\frac{d_0-1}{\delta d_0-1}\cdot|U|$.
    \end{itemize}
\end{lemma}

\begin{proof}
As $|F|\le \alpha n$, by \cref{lem:easyflip}, there exists $m_1\in[c]$ such that $x^1={\rm EasyFlip}(x,m_1)$ satisfies $$|F^1|\le(1-\epsilon_3)|F|\le\alpha n.$$
Continuing this process, it follows by \cref{lem:easyflip} that for each $k\in[s]$, there exists $m_k\in[c]$ such that $x^k={\rm EasyFlip}(x^{k-1},m_k)$ satisfies
\begin{align}\label{eq:F^k}
    |F^k|\le(1-\epsilon_3)|F^{k-1}|\le (1-\epsilon_3)^k|F|\le\alpha n.
\end{align}
Such a vector $(m_1,\ldots,m_s)\in[c]^s$ clearly satisfies property (i).

To prove (ii), note that it follows by \eqref{eq:unsatisfied_constraints} and \eqref{eq:F^k} that for each $k\in[s]$,
    \begin{align*}
        |U^k|\le |N(F^k)|\le c|F^k|\le(1-\epsilon_3)^k\cdot c|F|,
    \end{align*}
as needed.

To prove (iii), as $|F|\le\alpha n$, applying \cref{prop:expander} in concert with \eqref{eq:unsatisfied_constraints} gives that
    \begin{align*}
        \frac{\delta d_0-1}{d_0-1}\cdot c|F|\le |N_{\le d_0-1}(F)|\le|U|.
    \end{align*}
Combining the equation above and (i) gives that
    \begin{align*}
        |U^s|\le c|F^s|\le(1-\epsilon_3)^s\cdot c|F|\le (1-\epsilon_3)^s\cdot\frac{d_0-1}{\delta d_0-1}\cdot|U|,
    \end{align*}
completing the proof of (iii).
\end{proof}

\cref{lem:deepflip-ideal} (i) indicates that there exists an ``ideal'' choice, say $(m^*_1,\ldots,m^*_s)\in[c]^s$, such that if $|F|\le \alpha n$, then after the execution of {\rm \hyperref[alg:easyflip]{EasyFlip}} iteratively for $s$ times (directed by $(m^*_1,\ldots,m^*_s)$), the number of corrupt variables in the final output $x^s$ is at most a $(1-\epsilon_3)^s$-fraction of the number of corrupt variables in the initial input $x^0=x$.

Unfortunately, in general, there is no way to compute the number of corrupt variables in the input and output of each execution of {\rm \hyperref[alg:easyflip]{EasyFlip}}. From this perspective, there is no easy way to \textit{explicitly} find the ideal $(m^*_1,\ldots,m^*_s)\in[c]^s$. However, \cref{lem:deepflip-ideal} (iii), which is a consequence of \cref{lem:deepflip-ideal} (i), essentially shows that if the number of corrupt variables reduces \textit{dramatically}, then the number of unsatisfied constraints also reduces \textit{significantly} - fortunately, it is clear that this quantity can be computed in linear time! The analysis of our deterministic decoding algorithm relies heavily on this observation.

The above discussion motivates the following definition.

\begin{definition}\label{def:M}
Given the input vector $x$ of {\rm \hyperref[alg:deepflip]{DeepFlip}}, let $M$ be the set consisting of all vectors $(m_1,\ldots,m_s)\in[c]^s$ which satisfy the following two properties:
\begin{itemize}
    \item [{\rm (a)}] for each $k\in[s]$, $|U^k|\le (1-\epsilon_3)^k\cdot c\gamma n$;
    \item [{\rm (b)}] $|U^s|\le\epsilon_4 |U|$, where $\epsilon_4=\frac{\delta d_0-1}{d_0-1}\cdot(1-\epsilon_3)$.
\end{itemize}
\end{definition}

The following result is an easy consequence of \cref{lem:deepflip-ideal}.

\begin{claim}\label{lem:mid01}
If $|F|\le \gamma n$ and $s\ge s_0$, then $M\neq\emptyset$.
\end{claim}

\begin{proof}
Since $|F|\le \gamma n<\alpha n$, there exists a vector $(m_1,\ldots,m_s)\in[c]^s$ that satisfies \cref{lem:deepflip-ideal}. By substituting $|F|\le\gamma n$ into \cref{lem:deepflip-ideal} (ii), it is easy to see that such a vector also satisfies \cref{def:M} (a). Moreover, by substituting $s\ge s_0=\left\lceil\log_{1-\epsilon_3}\left(\epsilon_4\frac{\delta d_0-1}{d_0-1}\right)\right\rceil$ into \cref{lem:deepflip-ideal} (iii), it is not hard to see that \cref{def:M} (b) also holds. Therefore, $M\neq\emptyset$, as needed.
\end{proof}

As briefly mentioned above, in general one cannot explicitly find the ideal $(m^*_1,\ldots,m^*_s)\in[c]^s$ which dramatically reduces the number of corruptions. Instead, under a stronger condition $|F|\le\gamma n$ (recall that \cref{lem:deepflip-ideal} assumes $|F|\le\alpha n$), \cref{lem:deepflip} shows that for every $(m_1,\ldots,m_s)\in M$, $x^s={\rm DeepFlip}(x,(m_1,\ldots,m_s))$ reduces the number of corrupt variables of $x$ by a $(1-\epsilon_3)$-fraction, which makes every member of $M$ an \textit{acceptable} (which may be not ideal) choice for \hyperref[alg:deepflip]{DeepFlip}.

Now we are ready to present the proof of \cref{lem:deepflip}.

\begin{proof}[Proof of \cref{lem:deepflip}]
First of all, we would like to show that \cref{lem:deepflip} is well defined, namely, for every $|F|\le\gamma n$ and $(m_1,\ldots,m_s)\in M$, ${\rm DeepFlip}(x,(m_1,\ldots,m_s))$ does not return $\bot$. Indeed, as $(m_1,\ldots,m_s)\in M$, by \cref{def:M} (a) we have that for every $1\le k\le s$, $|U^{k}|\le (1-\epsilon_3)^{k}\cdot c\gamma n$, which implies that $U^{k}$ always passes the test in step 5 of \cref{alg:deepflip}. Therefore, under the assumption of \cref{lem:deepflip}, the output of \hyperref[alg:deepflip]{DeepFlip} is a vector $x^s\in\mathbb{F}_2^n$.

To prove the lemma, assume for the moment that $|F^s|\le\alpha n$. Given the correctness of this assertion, applying \cref{prop:expander} in concert with \eqref{eq:unsatisfied_constraints} gives that
    \begin{align*}
        \frac{\delta d_0-1}{d_0-1}\cdot c|F^s|\le |N_{\le d_0-1}(F^s)|\le|U^s|.
    \end{align*}
Moreover, by combining the above equation and \cref{def:M} (b), we have
\begin{align*}
\frac{\delta d_0-1}{d_0-1}\cdot c|F^s|\le|U^s|\le\epsilon_4|U|\le \frac{\delta d_0-1}{d_0-1}\cdot(1-\epsilon_3)\cdot c|F|,
\end{align*}
which implies that
\begin{align*}
    |F^s|\le (1-\epsilon_3)|F|,
\end{align*}
as needed.

Therefore, it remains to show that $|F^s|\le\alpha n$. We will prove by induction that for each $0\le k\le s$, $|F^k|\le\frac{\alpha n}{1+c/(d_0-t)}\le\alpha n$. For the base case $k=0$, it follows by assumption that $|F^0|\le\gamma n<\frac{\alpha n}{1+c/(d_0-t)}$ as $d_0\ge3$, $\delta d_0>3$ and $t=\lfloor\frac{1}{\delta}\rfloor$. Suppose that for some $k\in[s]$ we have $|F^{k-1}|\le\frac{\alpha n}{1+c/(d_0-t)}$. Since $x^k=${\rm EasyFlip}$(x^{k-1},m_k)$, it follows by \cref{lem:inner03} that
\begin{align*}
    |F^k|\le (1+\frac{c}{d_0-t})|F^{k-1}|\le\alpha n.
\end{align*}
Therefore, we have
\begin{align*}
    \frac{\delta d_0-1}{d_0-1}\cdot c|F^k|\le|N_{\le d_0-1}(F^k)|\le|U^k|\le(1-\epsilon_3)^k\cdot c\gamma n\le c\gamma n,
\end{align*}
where the first inequality follows from \cref{prop:expander}, the second inequality follows from \eqref{eq:unsatisfied_constraints}, and the third inequality follows from \cref{def:M} (a). The last equation implies that
\begin{align*}
    |F^k|\le\frac{d_0-1}{\delta d_0-1}\gamma n<\frac{d_0}{2}\gamma n\le\frac{\alpha n}{1+c/(d_0-t)},
\end{align*}
as needed, where the second inequality follows from the assumption $\delta d_0>3$ and the last inequality follows from the definition of $\gamma$ in \cref{thm:main}, $\delta d_0>3$ and $t=\lfloor\frac{1}{\delta}\rfloor$.

The proof of the lemma is thus completed.
\end{proof}

\subsubsection{Proof of \texorpdfstring{\cref{lem:deepflip-time}}{}}

\hyperref[alg:deepflip]{DeepFlip} essentially consists of $s$ \hyperref[alg:easyflip]{EasyFlip} invocations, which compute $x^k:={\rm EasyFlip}(x^{k-1},m_k)$ for all $k\in[s]$. So, it suffices to analyze their total running time.

We will need the vectors $\{z^u:u\in R\}$ and the counters $\{\tau_v:v\in L\}$ defined in the proof of \cref{lem:easyflip-time}. We will also compute $U^k=\{u\in R:x^k_{N(u)}\notin C_0\}$ for all $0\le k\le s$ and associate it with an indicator vector $\lambda\in\{0,1\}^{|R|}$ to record the set of unsatisfied constraints\footnote{Assume that the coordinates of $\lambda$ are labeled by constraints in $R$.}. Note that all of $\{z^u:u\in R\}$, $\{\tau_v:v\in L\}$, and $\lambda$ will be updated with the computation of $x^k$ during \hyperref[alg:deepflip]{DeepFlip}.

By \cref{lem:easyflip-time}, computing $x^1:={\rm EasyFlip}(x^{0},m_1)$ takes time $O(n+|F^0|)$. Given $x^{k-1}:={\rm EasyFlip}(x^{k-2},m_{k-1})$, let us analyze the running time of computing $x^k:={\rm EasyFlip}(x^{k-1},m_k)$, where $k\in\{2,\ldots,s\}$. Recall that we divided \hyperref[alg:easyflip]{EasyFlip} roughly into two parts (see the discussion above \cref{alg:easyflip}). We will compute the running time of these two parts separately.

\paragraph{EasyFlip (i):}
\begin{itemize}
    \item Note that we know $S_{m_{k-1}}$, which is the set of flipped variables in the computation of $x^{k-1}$. Therefore, to compute $x^k$ we only need to invoke Decode$(x^{k-1}_{N(u)})$ and compute $d_H({\rm Decode}(x^{k-1}_{N(u)}),x^{k-1}_{N(u)})$ for those $u\in N(S_{m_{k-1}})$, since these are the only constraints that could possibly see a status change after computing $x^{k-1}$. This process takes at most $O((t_0+d)|N(S_{m_{k-1}})|)$ time over all the constraints in $N(S_{m_{k-1}})$.

    \item Note that we also know the current value of $z^u$ for every $u\in R$, which indexes the neighbor of $u$ that receives the flip sent from $u$ in the computation of $x^{k-1}$. Note that $z^u$ has at most one nonzero coordinate. Now, to compute $x^k$, if a constraint $u\in N(S_{m_{k-1}})$ sends a flip to a variable $v\in L$, then we update $z^u$ by setting $z^u_v:=1$ and all other coordinates $0$. This process takes at most $O(|N(S_{m_{k-1}})|)$ time.

    \item In total, computing $x^k$ EasyFlip (i) takes $O((t_0+d)|N(S_{m_{k-1}})|)$ time.
\end{itemize}

\paragraph{EasyFlip (ii):}
\begin{itemize}
    \item We know the current value of $\tau_v$ for every $v\in L$, which counts the number of flips received by $v$ in the computation of $x^{k-1}$. It is not hard to see that to compute $x^k$, a counter $\tau_v$ may change only if $v$ is a neighbor of the constraints in $N(S_{m_{k-1}})$, as other constraints would not send any flip. Since the constraints in $R$ send at most $|N(S_{m_{k-1}})|$ flips to variables in $L$, at most $|N(S_{m_{k-1}})|$ variables in $L$ can receive flip messages. So, updating the counters for all $v\in L$ takes at most $O(c|N(S_{m_{k-1}})|)$ time.

    \item Lastly, we need to flip $|S_{m_k}|$ variables, which needs $O(|S_{m_k}|)$ time.

    \item In total, computing $x^k$ EasyFlip (ii) takes $O(c|N(S_{m_{k-1}})|+|S_{m_k}|)$ time.
\end{itemize}

To sum up, the running time of computing $x^k$ is at most
\begin{align*}
O((t_0+d)|N(S_{m_{k-1}})|)+O(c|N(S_{m_{k-1}})|+|S_{m_k}|)=O((t_0+d+c)|N(S_{m_{k-1}})|+|S_{m_{k}}|).
\end{align*}

\paragraph{Updating $U^k$.} Note that in \hyperref[alg:deepflip]{DeepFlip} (see step 5 in \cref{alg:deepflip}) we also need to compute $|U^k|$ for each $k\in[s]$. It is clear that $U^0$ can be found by invoking $\text{Check}(x^0_{N(u)})$ for all $u\in R$, which takes time $O(h_0|R|)$. Then we need to initialize $\lambda$ to a vector that represents $U^0$, which takes time $O(|R|)=O(n)$. Moreover, given $\lambda$ which represents $U^{k-1}$, to compute $U^k$ we only need to invoke Check$(x^{k}_{N(u)})$ for those $u\in N(S_{m_{k}})$, since these are the only constraints that could possibly see a status change after computing $x^{k}$. Therefore, for each $k\in[s]$, the running time of computing $|U^k|$ is at most $O(h_0|N(S_{m_{k}})|)$.

\paragraph{Running time of DeepFlip.} Noting that $F^0=F$, the running time of \hyperref[alg:deepflip]{DeepFlip} is at most
\begin{align*}
    &O(n+|F|)+\sum_{k=1}^s O(|S_{m_{k}}|+|N(S_{m_k})|)\le O(n+|F|)+\sum_{k=1}^s O(|S_{m_{k}}|)\\
\le &O(n+|F|)+\sum_{k=1}^{s} O(|F^k|)\le O(n+|F|)+\sum_{k=1}^{s} O((1+\frac{c}{d_0-t})^k|F|)\\
= &O(n+|F|),
\end{align*}
where the three inequalities follow from $|N(S_{m_k})|\le c|S_{m_k}|$,  \eqref{ineq:|S_m|} and \cref{lem:inner03}, respectively.

\subsection{Running DeepFlip thoroughly until significantly reducing the number of unsatisfied constraints -- HardSearch}\label{sec:hardsearch}

In this subsection, we describe and analyze \hyperref[alg:hardsearch]{HardSearch} (see \cref{alg:hardsearch} below). Given an input vector $x\in\mathbb{F}_2^n$ with at most $\gamma n$ corruptions, \hyperref[alg:hardsearch]{HardSearch} runs DeepFlip$(x,(m_1,\ldots,m_s))$ over all choices of $(m_1,\ldots,m_s)\in[c]^s$ until it finds one, say $(m'_1,\ldots,m'_s)$, such that the number of unsatisfied constraints with respect to ${\rm DeepFlip}(x,(m'_1,\ldots,m'_s))$ is at most an $\epsilon_4$-fraction of the number of unsatisfied constraints with respect to $x$. Then \cref{lem:main} shows that the number of corruptions in $x'$ is at most a $(1-\epsilon_3)$-fraction of the number of corruptions in $x$. Therefore, running \hyperref[alg:hardsearch]{HardSearch} iteratively for $\ell$ rounds gives us a $(1-\epsilon_3)^{\ell}$-reduction on the number of corruptions.

\begin{algorithm}[htbp]
\caption{Running DeepFlip over all $(m_1,\ldots,m_s)\in[c]^s$ until finding an ``acceptable'' one -- HardSearch}
\label{alg:hardsearch}
\begin{algorithmic}[1]
\REQUIRE $G$, $C_0$, $x\in\mathbbm{F}_2^n$, and $s=s_0$
\ENSURE $x'\in \mathbbm{F}_2^n$
\STATE $U\gets \{u\in R: x_{N(u)}\notin C_0\}$
\FOR{every $(m_1,\ldots,m_s)\in [c]^s$}
\STATE $x'\gets$ DeepFlip$(x,(m_1,\ldots,m_s))$
\IF{$x'\neq \bot$ }
\STATE $U'\gets \{u\in R: x'_{N(u)}\notin C_0\}$
\IF{$|U'|\le \epsilon_4 |U|$}
\RETURN $x'$
\ENDIF
\ENDIF
\ENDFOR
\end{algorithmic}
\end{algorithm}

\subsubsection{Proof of \texorpdfstring{\cref{lem:main}}{}}
To prove (i), let $M$ be the set of vectors in $[c]^s$ which satisfy the two conditions in \cref{def:M} with respect to $F$ and $s$, where $|F|\le\gamma n$ and $s=s_0$. By our choices of $F$ and $s$, it follows by \cref{lem:mid01} that $M\neq\emptyset$. By \cref{lem:deepflip}, as long as \hyperref[alg:hardsearch]{HardSearch} finds a vector $(m_1,\ldots,m_s)\in M$, it would output a vector
$x'={\rm DeepFlip}(x,(m_1,\ldots,m_s))$ such that $|U'|\le \epsilon_4 |U|$\footnote{This holds since $(m_1,\ldots,m_s)\in M$ satisfies \cref{def:M} (b).} and $|F'|\le(1-\epsilon_3)|F|$, as needed.

It remains to prove (ii), which is an easy consequence of (i). Let $F^i$ be the set of corruptions in $x^i$ for all $0\le i\le\ell$. Then by (i) for every $0\le i\le \ell-1$, we have either $x^{i}\in T(G,C_0)$ (if $|U^i|=0$) or $|F^{i+1}|\le(1-\epsilon_3)|F^i|$ (if $|U^i|\neq 0$). Therefore, after at most $\ell=\left\lceil\log_{1-\epsilon_3}\left(\left\lfloor\frac{d_0-1}{2}\right\rfloor\frac{1}{\gamma n}\right)\right\rceil$ iterative executions of \hyperref[alg:hardsearch]{HardSearch}, the number of corrupt variables is at most
\begin{align*}
    (1-\epsilon_3)^\ell\gamma n\le \left\lfloor\frac{d_0-1}{2}\right\rfloor\frac{1}{\gamma n}\cdot \gamma n=\left\lfloor\frac{d_0-1}{2}\right\rfloor,
\end{align*}
as needed.

\subsubsection{Proof of \texorpdfstring{\cref{lem:HardSearchtime}}{}}

\paragraph{Proof of \cref{lem:HardSearchtime} (i).}
Fix an arbitrary total order on $[c]^s$, and assume that \hyperref[alg:hardsearch]{HardSearch} goes through all vectors in $[c]^s$ increasingly in this order. Note that in the worst case, \hyperref[alg:hardsearch]{HardSearch} could invoke \hyperref[alg:deepflip]{DeepFlip} for $c^s$ times. We will estimate the running time of \hyperref[alg:hardsearch]{HardSearch} for this worst case.

Note that the input vectors of all \hyperref[alg:deepflip]{DeepFlip} invocations are the same, i.e., equal to $x$, and we do not want to write down the entire vector $x$ from time to time, as it could take a lot of time. Instead, we will use a binary vector $w\in\mathbb{F}_2^n$ to record the flipped variables during each invocation of \hyperref[alg:deepflip]{DeepFlip}. 
If the current choice $(m_1,\ldots,m_s)\in[c]^s$ does not satisfy step 6 of \hyperref[alg:hardsearch]{HardSearch}, then we can use $w$ together with the current value of $x'$ to recover $x$ (indeed, $x=x'+w$) and then turn to a new choice of  $(m_1,\ldots,m_s)\in[c]^s$.

We will show that updating $w$ in any \hyperref[alg:deepflip]{DeepFlip} invocation costs only $O(|F|)$ time, where $F$ is the set of corrupt variables of $x$. Indeed, we initially set $w=0^n$ and if a variable $v\in L$ is flipped, then $w_v$ increases by 1 (addition modulo 2). Note that ${\rm DeepFlip}(x,(m_1,\ldots,m_s))$ flips a total number of $\sum_{k=1}^s|S_{m_k}|$ variables. For each $k\in[s]$, let $F^k$ be the set of corrupt variables of ${\rm DeepFlip}(x,(m_1,\ldots,m_k))$. It follows that both ${\rm wt}(w)$ and the time of updating $w$ can be bounded from the above by
    \begin{align}\label{ineq:Hardtime01}
    O(\sum_{k=1}^s|S_{m_k}|)=O(\sum_{k=0}^{s-1}|F^k|)=O(\sum_{k=0}^{s-1}(1+\frac{c}{d_0-t})^k|F|)=O(|F|),
    \end{align}
where the second and third inequalities follow from \eqref{ineq:|S_m|} and \cref{lem:inner03}, respectively.

Observe that \hyperref[alg:hardsearch]{HardSearch} contains at most $c^s$ \hyperref[alg:deepflip]{DeepFlip} invocations and each \hyperref[alg:deepflip]{DeepFlip} contains $s$ \hyperref[alg:easyflip]{EasyFlip} invocations. For notational convenience, we use the pair $(j,k)\in[c^s]\times[s] $ to index the $(j,k)$-th \hyperref[alg:easyflip]{EasyFlip} in \hyperref[alg:hardsearch]{HardSearch}. 
\begin{itemize}
    \item In the 1st \hyperref[alg:deepflip]{DeepFlip}, initializing and updating $w$ take $O(n+|F|)$ time. Therefore, by \cref{lem:easyflip-time}, the running time of the 1st \hyperref[alg:deepflip]{DeepFlip} is $O(n+|F|)$.
    \item For each $2\le j\le c^s$, we will compute the running time of the $j$-th \hyperref[alg:deepflip]{DeepFlip} as follows.
    \begin{itemize}
        \item First, we will use $w$ to recover the input vector $x$ of \hyperref[alg:hardsearch]{HardSearch}. More precisely, we flip all $x_v$ with $w_v=1$, which takes $O({\rm wt}(w))=O(|F|)$ time, as shown by \eqref{ineq:Hardtime01}.
        \item Note that the $(j,1)$-th \hyperref[alg:easyflip]{EasyFlip} only needs to invoke Decode for the neighbors of the variables $\{v\in L:w_v=1\}$, as the status of the other constraints would not change. After that, we reset $w$ to $0^n$. The above process altogether takes $O(|F|)$ time.
        \item For $2\le k\le s$, the analysis of the running time of the $(j,k)$-th \hyperref[alg:easyflip]{EasyFlip} is similar to the one presented in the proof \cref{lem:deepflip-time}, which takes $O(|F|)$ time.
        \item Note that $|U|$ and $|U'|$ are already known in each \hyperref[alg:deepflip]{DeepFlip} invocation, as shown in the proof of \cref{lem:deepflip-time}.
    \end{itemize}
      In summary, $j$-th \hyperref[alg:deepflip]{DeepFlip} requires $O(|F|)$ time.
\end{itemize}
Therefore, the running time of \hyperref[alg:hardsearch]{HardSearch} is $O(n+|F|)$.

\paragraph{Proof of \cref{lem:HardSearchtime} (ii).}
Let $x^0=x$ be the input vector of \hyperref[alg:main]{MainDecode} and $F^0$ be the set of corrupt variables of $x^0$. Note that \hyperref[alg:main]{MainDecode} contains $\ell$ \hyperref[alg:hardsearch]{HardSearch} invocations. For every $i\in[\ell]$, let $x^i:=$HardSearch$(x^{i-1})$ and $F^i$ be the set of corrupt variables in $x^i$. Moreover, similarly to the notation above, we use the triple $(i,j,k)\in[\ell]\times[c^s]\times[s]$ to index the $(i,j,k)$-th \hyperref[alg:easyflip]{EasyFlip} and the pair $(i,j)\in [\ell]\times[c^s]$ to index the $(i,j)$-th \hyperref[alg:deepflip]{DeepFlip} in \hyperref[alg:main]{MainDecode}.
\begin{itemize}
    \item By \cref{lem:HardSearchtime} (i), computing $x^1:=$HardSearch$(x^{0})$ takes $O(n+|F^0|)$ time.
    \item For each $2\le i\le \ell$, we analyze the running time of computing $x^i:=$HardSearch$(x^{i-1})$ as follows.
    \begin{itemize}
        \item First, we need to reset the vector $w$, which records the flipped variables of $x^{i-2}$ in computing $x^{i-1}$, to $0^n$. By \eqref{ineq:Hardtime01} the above process takes time $O({\rm wt}(w))=O(|F^{i-2}|)$. We also need to update $w$ during the calculation of $x^i$, which takes time $O(|F^{i-1}|)$ by \eqref{ineq:Hardtime01}.
        \item Note that the $(i,1,1)$-th \hyperref[alg:easyflip]{EasyFlip} only needs to invoke Decode for the neighbors of the variables that are flipped in the $(i-1,c^s,s)$-th \hyperref[alg:easyflip]{EasyFlip}, which takes time $O(|F^{i-2}|)$, since by \eqref{ineq:Hardtime01} the size of the neighbors of those variables are bounded by $O(|F^{i-2}|)$. After that, by \cref{lem:easyflip-time} we need an extra $O(|F^{i-1}|)$ time to finish this \hyperref[alg:easyflip]{EasyFlip}. For $2\le k\le s$, the analysis of the running time of the $(i,1,k)$-th \hyperref[alg:easyflip]{EasyFlip} is similar to the one presented in the proof of \cref{lem:deepflip-time}, which is at most $O(|F^{i-1}|)$. Therefore, the $(i,1)$-th \hyperref[alg:deepflip]{DeepFlip} takes $O(|F^{i-2}|+|F^{i-1}|)$ time.
        \item For $2\le j\le c^s$, similar to the discussion in the proof of (i), the running time of the $(i,j)$-th \hyperref[alg:deepflip]{DeepFlip} takes $O(|F^{i-1}|)$ time.
        \item The set $U^{i-1}$ has already been computed in the $(i-1)$-th \hyperref[alg:hardsearch]{HardSearch}.
    \end{itemize}
    Therefore, the $i$-th \hyperref[alg:hardsearch]{HardSearch} takes time $O(|F^{i-2}|+|F^{i-1}|)$.
\end{itemize}

To sum up, the total running time of \hyperref[alg:main]{MainDecode} is
    \begin{align*}
         &O(n+|F^0|)+\sum\limits_{i=2}^\ell O(|F^{i-2}|+|F^{i-1}|)= O(n+|F^0|)+\sum\limits_{i=0}^{\ell-1}O(|F^i|)\\
         =& O(n+|F^0|)+\sum\limits_{i=0}^{\ell-1}O((1-\epsilon_3)^i|F^0|)=  O(n),
    \end{align*}
where the second equality follows from \cref{lem:main}.

\section{Randomized decoding: Reduce large corruptions to a moderate size} \label{sec:rand_decoding}

\noindent In this section, we present our randomized decoding for Tanner codes which can correct more errors. The general strategy is as follows. First, we use a voting process to derive a set $S:=\cup_{m\in [c]} S_{m}$ of candidate variables to flip. More precisely, each constraint $u\in R$ that satisfies $1\le d_H({\rm Decode}(x_{N(u)}),x_{N(u)})\le t$ sends exactly one flip to an arbitrary variable $v\in N(u)$ with ${\rm Decode}(x_{N(u)})_v\neq x_v$. We then design a special sampling process to pick a large fraction of variables from $S$ and flip them. This process can, with high probability, reduce the number of corrupted variables by a positive fraction. We repeat the above random process until the number of corrupted variables drops below $\gamma n$, in which case our deterministic decoding \hyperref[alg:main]{MainDecode} in \cref{alg:main} can work correctly, or we run out of time and stop. Finally, we use \hyperref[alg:main]{MainDecode} to get the codeword.

Let $\gamma$ be the relative decoding radius of Theorem \cref{thm:main}. The exact randomized decoding is given as \cref{algo:rand}, which yields the following result.

\begin{algorithm}[htbp]
\caption{Randomized Decoding}
\begin{algorithmic}[1]\label{algo:rand}
\REQUIRE  $x\in\mathbbm{F}_2^n$ with at most $\alpha $ fraction errors
\ENSURE a codeword in $T(G, C_0)$ or $\bot$
\STATE   Set $t = \lfloor \frac{1}{\delta} \rfloor$
\FOR{ $\ell = 1,\ldots,\left\lceil \frac{\log \frac{\gamma }{\alpha } }{ \log \left(1-  \frac{3\epsilon(\delta(t+1)-1)}{4t} \right) } \right\rceil$}
\FOR{every $u\in R$}
\STATE $\omega\gets \text{Decode}(x_{N(u)})$
\IF{$1\le d_H(\omega,x_{N(u)})\le t$}
\STATE send a ``flip'' message to the vertex $v\in N(u)$ with the smallest index such that $\omega_v\neq x_v$
\ENDIF
\ENDFOR
\STATE $\forall m\in [c], S_m \gets \{v\in L: v \mbox{ receives } m \mbox{ ``flip'' messages} \}$
\STATE $S \gets \bigcup_{m}S_m$
\STATE Randomly pick $P \subseteq S$: for every $m\in [c]$, for each variable in $S_m$,  pick it with probability $ \frac{m}{2 c} $, using independent randomness
\STATE Flip all bits in $P$
\STATE $U \gets \{u\in R: x_{N(u)}\notin C_0\}$
\IF{$|U| \le \left( \delta - \frac{1}{d_0}\right) c\gamma n $}
\RETURN MainDecode($x$) 
\ENDIF
\ENDFOR
\RETURN $\bot$
\end{algorithmic}
\end{algorithm}

\begin{theorem}[restatement of \cref{thm:randDec}]
    Let $G$ be a $(c,d,\alpha,\delta)$-bipartite expander and $C_0$ be a $[d,k_0,d_0]$-linear code, where $c,d,\alpha,\delta,d_0,k_0$ are positive constants. If $\delta d_0 > 3$, then there exists a linear-time randomized decoding algorithm for Tanner code $T(G,C_0)$ such that if the input has at most $\alpha n$ errors from a codeword, then with probability $1-\exp\left\{  -\Theta_{c,\delta, d_0}\left( n \right) \right\}$, the decoding algorithm can output the correct codeword.
\end{theorem}

 The following two lemmas demonstrate the correctness and linear running time of \cref{algo:rand}, respectively.

\begin{lemma}\label{lem:randDecCorrect}
    If the input has distance at most $\alpha n$ from a codeword, then with probability $1-\exp\left\{ - \Theta_{c,\delta,d_0}(n) \right\}$, \cref{algo:rand} outputs the correct codeword.

\end{lemma}

\begin{lemma}\label{lem:randDecTime}
    If the input has distance at most $\alpha n$ from a codeword, then \cref{algo:rand} runs in linear time.
\end{lemma}

Assuming the correctness of the above lemmas, we can prove \cref{thm:randDec} as follows.

\begin{proof}[Proof of \cref{thm:randDec}]
If the input word has at most $\alpha n$ errors, then by \cref{lem:randDecCorrect}, with probability $1-\exp\left\{ -\Theta_{c, \delta, d_0}(n) \right\}$, the decoding outputs the correct codeword. Furthermore, the running time is linear by \cref{lem:randDecTime}.
\end{proof}

\subsection{Proof of \texorpdfstring{\cref{lem:randDecCorrect} and \cref{lem:randDecTime}}{}}
Recall that we defined $A$ as the set of constraints $u\in R$ that sends a flip and \text{Decode}$(x_{N(u)})$ computes the correct codeword in $C_0$ ( see \eqref{eq:def-A}), and $B$ to be the set of constraints $u\in R$ that sends a flip and \text{Decode}$(x_{N(u)})$ computes an incorrect codeword in $C_0$ (see \eqref{eq:def-B} ). Also, recall we let $\alpha_m$ denote the fraction of corrupt variables in $S_m$ (see \eqref{eq:def-alpha}) and we let $\beta_m$ denote the fraction of flips sent from $A$ to $S_m$ among all flips received by $S_m$ (see \eqref{eq:def-beta}). Now we let $ M:= A \cup B$.

First, we bound the size of $P$ in an arbitrary iteration.
\begin{claim}\label{lem:sizeofP}
For every constant $\epsilon > 0$, with probability $\ge 1 - \exp\left\{-\Theta_{c, \epsilon}(|M|)\right\}$, the size of $P$ is in $ \left[(1-\epsilon)\frac{|M|}{2c} , (1+\epsilon) \frac{|M|}{2c}  \right] $.
\end{claim}
\begin{proof}
    In Algorithm \ref{algo:rand}, for every $m\in [c]$, each variable in $S_i$ is picked independently with probability $\frac{m}{2c}$. For each $v\in S$, let $X_v$ be the indicator random variable of the event that the variable $v$ is picked. So for every $v\in S_m$, $\Pr[X_v = 1] = \frac{m}{2c}$. Let $X = \sum_{v\in S} X_v$. It is easy to see that $X = |P|$. By the linearity of expectation, we have that
    $$ \text{E} X = \sum\limits_{v \in S} \text{E} X_v =  \sum\limits_{m\in [c]} \frac{m}{2c}|S_m| = \frac{|M|}{2c}.$$
    By Hoeffding's inequality,
    $$\Pr\left[ X  \in  \left[(1-\epsilon)\frac{|M|}{2c} , (1+\epsilon) \frac{|M|}{2c}  \right]  \right] \ge 1- 2  \exp\left\{-\frac{2 \left(\epsilon \frac{|M|}{2c}\right)^2}{|S|}  \right\}
 \ge 1- 2\exp\left\{-\frac{\epsilon^2|M|}{2c^2}  \right\},$$
 where the second inequality follows from the fact that $|S| \le |M|$.
\end{proof}

Next, we show that $P$ contains significantly more corrupted variables than corrected variables.

\begin{claim}\label{lem:corruptionsizeinP}
There exists a constant $\epsilon  $ such that with probability $\ge 1 - \exp\left\{-\Theta_{c, \epsilon}(|M|)\right\} $, the number of corrupted variables in $P$ is at least $ (1/2+\epsilon)  \frac{|M|}{2c}$.
\end{claim}

\begin{proof}[Proof of \cref{lem:corruptionsizeinP}]
For every $v\in S$, let $Y_v$ be the indicator random variable of the event that $X_v=1$ and $v\in F$. Let $Y=\sum_{v\in S } Y_v$. By definition, $Y=|P\cap F|$. Note that for every $v\notin S\cap F$, $\Pr\left[Y_v = 1 \right] = 0$. By the linearity of expectation, we have that
\begin{equation}
 \text{E} Y = \sum_{v\in S}\text{E}  Y_v =\sum_{m\in [c]} \sum_{v\in S_m}\text{E} Y_v =  \sum_{m\in [c]} \sum_{v\in S_m\cap F}  \frac{m}{2c}=\sum_{m\in [c]}\frac{m}{2c}\alpha_m|S_m| \ge \sum_{m\in [c]}   \frac{m}{2c}\beta_m |S_m|,
\end{equation}
where the inequality follows from \cref{lem:alpha}.

By the definition of $\beta_m$,
\[ m\beta_m |S_m| =  \left| \left\{ \mbox{ The number of ``flips'' sent from } A \mbox{ to } S_m  \right\} \right|.\]
Hence, one can infer that
\begin{align*}
\text{E} Y  &\ge \sum_{m\in [c]} \frac{ \left| \left\{ \mbox{ The number of ``flips'' sent from } A \mbox{ to } S_m  \right\} \right|}{2c}\\
       & = \frac{\left| \left\{ \mbox{ The number of ``flips'' sent from } A \mbox{ to } S   \right\} \right|}{2c} \\
       & = \frac{\left|A\right|}{2c},
\end{align*}
where the last equality is due to that each constraint in $R$ can only send at most $1$ message.

Set $\epsilon=\frac{\epsilon_0\delta^2}{4}>0$. It follows by \eqref{ineq:|A|/(|A|+|B|)} that
\begin{align*}
   |A|\ge \left(\frac{1}{2}+\frac{\epsilon_0\delta^2}{2}\right)|M|=\left(\frac{1}{2}+2\epsilon\right)|M|.
\end{align*}
Thus, we can infer that
\[ \text{E} Y \ge   \frac{\left|  A \right|}{2c} \ge \left(\frac{1}{2}  + 2\epsilon \right) \frac{|M|}{2c}.\]

By Hoeffding's inequality,
$$\Pr\left[ Y  \le  \left( \frac{1}{2}  + \epsilon \right) \frac{|M|}{2c}   \right] \ge 1-  \exp\left\{-\frac{2 \left(\epsilon \frac{|M|}{2c}\right)^2}{|S|}  \right\} \ge 1-  \exp\left\{-\frac{\epsilon^2|M|}{2c^2}  \right\},$$
where the second inequality follows from that $|S| \le |M|$.
\end{proof}

The following claim shows that as long as the number of unsatisfied constraints is small enough, we can ensure that the number of corrupt variables is at most $\gamma n$. Hence, we can handle the matter with \cref{alg:main}.

\begin{claim}\label{lem:callDeterministicDec}
    If $|U| \le   \left( \delta - \frac{1}{d_0}\right) c\gamma n   $ and $|F| \le \alpha n$, then $ |F| \le \gamma n $.
\end{claim}

\begin{proof}
Suppose that $ \gamma n<|F|  \le \alpha n$. By \cref{prop:expander}, we have that
$$|U|\ge\frac{\delta d_0-1}{d_0-1}c|F|>\left(\delta -\frac{1}{d_0}\right)c\gamma n,$$
which is a contradiction.
\end{proof}

Now, we can give the proofs of \cref{lem:randDecCorrect} and \cref{lem:randDecTime}, respectively, as follows.

\begin{proof}[Proof of \cref{lem:randDecCorrect}]
In each iteration, consider the case that the number of errors $|F|$ is at most $\alpha n$. If $|U| \le   \left( \delta - \frac{1}{d_0}\right) c\gamma n $, then by \cref{lem:callDeterministicDec}, $|F| \le \gamma n$. Therefore, it follows by \cref{thm:main} that when $\delta d_0 > 3$, all errors can be corrected. Otherwise, we claim that the number of corrupt variables can be decreased by a constant fraction in this iteration.

Recall that $M=A\cup B\subseteq U$. It follows by \eqref{ineq:|A|} that
\begin{align}\label{ineq:|M|}
 |M|=|A|+|B| \ge \frac{\delta(t+1)-1}{t}c|F|.
\end{align}
Note that by \cref{lem:corruptionsizeinP}, with probability $1-\exp\{-\Theta_{c,\delta,d_0}(n) \}$, the number of corruptions in $P$ is at least $(1/2+\epsilon)  \frac{|M|}{2c}$ where $\epsilon>0$ is a constant. Also note that by \cref{lem:sizeofP}, with probability $1-\exp\{-\Theta_{c,\delta,d_0}(n) \}$, the size of $P$ is in $ \left[(1-\epsilon/2)\frac{|M|}{2c} , (1+\epsilon/2) \frac{|M|}{2c}  \right] $. When both of these events occur, by flipping all variables in $P$, the number of corruptions is reduced by at least $\frac{3\epsilon |M|}{ 4
 c} $. It follows by \eqref{ineq:|M|} that
$$\frac{3\epsilon |M|}{4c}\ge\frac{3\epsilon(\delta(t+1)-1)}{4t}|F|.$$
This shows the number of corrupt variables indeed is decreased by a constant fraction in this iteration.

As a result, after at most $ \frac{\log \frac{\gamma }{\alpha } }{ \log \left(1-  \frac{3\epsilon(\delta(t+1)-1)}{4t} \right) }   $  iterations, the number of corruptions is at most $\gamma n$. Then the decoding can call \cref{alg:main} to correct all errors.
\end{proof}

\begin{proof}[Proof of \cref{lem:randDecTime}]
    Note that the number of iterations is constant. For each iteration, running the decoder of $C_0$ for all right vertices takes linear time. Sending messages and deriving $S_i, i\in [c]$ also take linear time, since each right vertex can send at most $1$ message. For picking $P$, notice that picking each variable uses a constant number of bits since the probability of picking $1 $ variable is a constant. Flipping bits in $P$ also takes linear time as $P \subseteq [n]$. By \cref{thm:main}, the \hyperref[alg:main]{MainDecode} process takes linear time. So, the overall running time is linear.
\end{proof}

\section{A lower bound on \texorpdfstring{$\delta d_0$}{}: Proof of \texorpdfstring{\cref{prop:lowerbd}}{}}

In this section we will prove the following result, which provides a necessary condition $\delta d_0>1$ for \cref{que-2}. Our proof borrows some ideas from the result of Viderman \cite{viderman2013linear}, which showed that $\delta>1/2$ is necessary for \cref{que-1}.

\begin{proposition}[restatement of \cref{prop:lowerbd}]
    For every $d,d_0\ge2$ and $n\ge 10d_0$, there exist constants $0<\alpha<1,c\ge3$ and a $(c,d,0.9\alpha,\frac{1}{d_0})$-bipartite expander $G$ with $V(G)=L\cup R$ and $|L|=n$ such that for every $[d,k_0,d_0]$-linear code $C_0$, the Tanner code $T(G,C_0)$ has minimum Hamming distance at most $d_0$.
\end{proposition}

To prove the proposition, we will construct an expander $G$ with expansion ratio $\delta=\frac{1}{d_0}$ such that for every $[d,k_0,d_0]$-linear code $C_0$, there exists a codeword $y\in T(G,C_0)$ with Hamming weight exactly $d_0$. In fact, it suffices to show that every constraint in $R$ is either adjacent to all variables in the support of $y$ or to none of the variables in the support of $y$.

\begin{proof}
The required graph $G$ is constructed as follows. We will ignore the floorings and ceilings whenever they are not important. Assume that $V(G)$ admits a bipartition with $V(G)=L\cup R$, where $L=[n]$. Let $L_1=[n-d_0]$, $L_2=L\backslash L_1=\{n-d_0+1,\ldots,n\}$, and $R=R_1\cup R_2\cup R_3$, where the $R_i$'s are pairwise disjoint, and are constructed as follows.
\begin{itemize}
    \item Given $d,d_0\ge 2$ and $n\ge 10d_0$, by \cref{prop:random}, there exist a constant $0<\alpha<1$ and an integer $c\ge 2$ such that there exists a $(c-1,d,\alpha,\frac{3}{4})$-bipartite expander $G_1$ with bipartition $V(G_1)=L_1\cup R_1$ and $|L_1|=n-d_0$.
    \item  Let $R_2=\{u_1,\ldots,u_c\}$ be a set of $c$ constraints, where for each $i\in [c]$,
$$N(u_i)=\{(i-1)(d-d_0)+1,\ldots,(i-1)(d-d_0)+d-d_0\}\cup L_2.$$
\item Set $m=\frac{n-d_0-(d-d_0)c}{d}$. Let $R_3=\{u'_1,\ldots,u'_m\}$ be a set of $m$ constraints, where for each $i\in[m]$,
$$N(u'_i)=\{(i-1)d+1+(d-d_0)c,\ldots,(i-1)d+d+(d-d_0)c\}.$$
\end{itemize}

It is routine to check that $G$ is $(c,d)$-regular. Next, we will show that every subset $S$ of $L$ with $|S|\le0.9\alpha n$ has at least $c|S|/d_0$ neighbors in $R$. Let $S_1=S\cap L_1$ and $S_2=S\setminus S_1$. It follows by $n\ge 10d_0$ that $|S_1|\le0.9\alpha n\le\alpha(n-d_0)=\alpha |L_1|$. Therefore, we have
    \begin{align*}
        |N(S_1)|\ge|N(S_1)\cap R_1|\ge \frac{3}{4}(c-1)|S_1|\ge\left(\frac{3}{4}-\frac{1}{4}\right)c|S_1|=\frac{1}{2}c|S_1|\ge \frac{1}{d_0}c|S_1|,
    \end{align*}
where the third inequality follows from $c\ge3$ and the last inequality follows from $d_0\ge2$. Moreover, by the definition of $G$, if $S_2\neq\emptyset$ then $|N(S_2)|=c$. Therefore, it follows by $|S_2|\le d_0$ that
    \begin{align*}
        |N(S_2)|\ge \frac{1}{d_0}c|S_2|.
    \end{align*}
Combining the two inequalities above, one can infer that
    \begin{align*}
        |N(S)|&=|N(S)\cap (R_1\cup  R_3)|+|N(S)\cap R_2|\ge |N(S_1)\cap R_1|+|N(S_2)\cap R_2|\\
        &\ge \frac{1}{d_0}c|S_1|+\frac{1}{d_0}c|S_2|=\frac{1}{d_0}c|S|,
    \end{align*}
which implies that $G$ is a $(c,d,0.9\alpha,\frac{1}{d_0})$-bipartite expander.

Let $C_0\subseteq\mathbb{F}_2^d$ be a $[d,k_0,d_0]$-linear code and assume without loss of generality that $0^{d-d_0}1^{d_0}\in C_0$. It suffices to show that $y:=0^{n-d_0}1^{d_0}\in T(G,C_0)$, which implies that $d(T(G,C_0))\le d_0$ as $T(G,C_0)$ is a linear code. Indeed, it is routine to check that for every $u\in R_1\cup R_3$, $y_{N(u)}=0^d\in C_0$ and for every $u\in R_2$, $y_{N(u)}=0^{d-d_0}1^{d_0}\in C_0$, as needed.
\end{proof}

\section*{Acknowledgements}
Minghui Ouyang would like to thank Extremal Combinatorics and Probability Group (ECOPRO), Institute for Basic Science (IBS, Daejeon, South Korea) for hosting his visit at the end of 2023.

The research of Yuanting Shen and Chong Shangguan is supported by the National Key Research and Development Program of China under Grant No. 2021YFA1001000, the National Natural Science Foundation of China under Grant Nos. 12101364 and 12231014, and the Natural Science Foundation of Shandong Province under Grant No. ZR2021QA005.

{\small
\bibliographystyle{plain}
\bibliography{Expander_codes}
}

\end{document}